\documentclass[12pt]{article}
\usepackage[english]{babel}
\usepackage[utf8]{inputenc}
\usepackage{lastpage}
\usepackage{graphicx}
\usepackage{epstopdf}
\usepackage{pgf,tikz}
\usepackage{mathrsfs}
\usepackage{subfigure}
\usepackage{algorithm}
\usepackage{mathtools}
\usepackage{amsthm}
\usepackage{amsmath}
\usepackage{amssymb}
\usepackage{amsfonts}
\usepackage{calrsfs}
\usepackage{caption}
\usepackage{hyperref}
\usepackage{mathrsfs}
\usepackage[autostyle=true]{csquotes}
\usepackage[noend]{algpseudocode}
\usepackage{algorithmicx}
\usepackage{pgfplots}

\usetikzlibrary{arrows}

\title{Lower Bounds on the Integraliy Ratio of the Subtour LP for the Traveling Salesman Problem}
\author{Xianghui Zhong\\
\small University of Bonn\\
            }
\date{\today} 

\newtheorem{thm}{Theorem}[section]
\theoremstyle{plain}
\newtheorem{lemma}[thm]{Lemma}
\newtheorem{theorem}[thm]{Theorem}
\newtheorem{corollary}[thm]{Corollary}

\newtheorem{conjecture}[thm]{Conjecture}
\theoremstyle{remark}
\newtheorem{remark}[thm]{Remark}
\newtheorem{observation}[thm]{Observation}
\theoremstyle{definition}
\newtheorem{definition}[thm]{Definition}
\newcommand*{\R}{\ensuremath{\mathbb{R}}}

\newcommand*{\N}{\ensuremath{\mathbb{N}}}

\DeclareMathOperator{\dist}{dist}
\DeclareMathOperator{\sgn}{sgn}
\DeclareMathOperator{\ratio}{ratio}
\DeclareMathOperator{\diffInner}{diffInner}
\DeclareMathOperator{\diffOuter}{diffOuter}

\algdef{SE}[DOWHILE]{Do}{doWhile}{\algorithmicdo}[1]{\algorithmicwhile\ #1}

\begin{document}

\maketitle
\begin{abstract}
In this paper we investigate instances with high integrality ratio of the subtour LP. We develop a procedure to generate families of \textsc{Euclidean TSP} instances whose integrality ratios converge to $\frac{4}{3}$ and may have a different structure than the instances currently known from the literature. Moreover, we compute the instances maximizing the integrality ratio for \textsc{Rectilinear TSP} with up to 10 vertices. Based on these instances we give families of instances whose integrality ratio converge to $\frac{4}{3}$ for \textsc{Rectilinear}, \textsc{Multidimensional Rectilinear} and \textsc{Euclidean} TSP that have similar structures. We show that our instances for \textsc{Multidimensional Rectilinear TSP} and the known instances for \textsc{Metric TSP} maximize the integrality ratio under certain assumptions. We also investigate the concept of local optimality with respect to integrality ratio and develop several algorithms to find instances with high integrality ratio. Furthermore, we describe a family of instances that are hard to solve in practice. The currently fastest TSP solver \texttt{Concorde} needs more than two days to solve an instance from the family with 52 vertices.
\end{abstract}

{\small\textbf{keywords:} traveling salesman problem; rectilinear TSP; Euclidean TSP; subtour LP; integrality ratio; exact TSP Solver}

\section{Introduction}
The traveling salesman problem (TSP) is probably the best-known problem in discrete optimization. An instance consists of the pairwise distances of $n$ vertices and the task is to find a shortest Hamilton cycle, i.e.\ a tour visiting every vertex exactly once. The problem is known to be NP-hard \cite{garey1979computers}. A special case of the \textsc{TSP} is the \textsc{Metric TSP}. Here the distances satisfy the triangle inequality. This \textsc{TSP} variant is still NP-hard  \cite{TSPNPHARD}. 

A well studied special case of the \textsc{Metric TSP} is the \textsc{Euclidean TSP}. Here an instance consists 
of points in the Euclidean plane and distances are defined by the 2-norm. 
The \textsc{Euclidean TSP} is still NP-hard~\cite{euclideannpg, euclideannpp}.

Instead of using the 2-norm as distance function for \textsc{Euclidean TSP}, one could also consider the 1-norm. This special case is called \textsc{Rectilinear TSP}. We can further generalize this special case by considering instances where the vertices and be embedded into the $d$-dimensional space instead of the Euclidean plane to get the \textsc{Multidimensional Rectilinear TSP}.

The Traveling Salesman Problem can be formulated as an integer program. One of the most common linear relaxation of the TSP is the \emph{subtour LP} \cite{dantzig1954solution}. For a given TSP instance defined by a complete graph $K_n$ and the cost function $c$ the subtour LP is given by:
\begin{align}
\min \sum_{e\in E(K_n)} &c(e)x_e \nonumber\\
\sum_{e \in \delta(v)} x_e &=2 && \text{for all~} v \in V(K_n) \label{degree constraint}\\
\sum_{e \in E(\delta(X))}x_e &\geq 2 && \text{for all~} \emptyset\subset X \subset V(K_n) \label{subtour elimination constraint}.\\ 
0 \leq x_e&\leq 1 && \text{for all~} e\in E(K_n) \nonumber
\end{align}

The constraints (\ref{degree constraint}) are called the \emph{degree constraints} and the constraints (\ref{subtour elimination constraint}) are called the \emph{subtour elimination constraints}. Although this LP has an exponential number of constraints it can be solved in polynomial time by the ellipsoid method since the separation problem can be solved efficiently \cite{grotschel1981ellipsoid}. A solution to the subtour LP is also called a \emph{fractional tour}.

Let $OPT(I)$ and $OPT_{LP}(I)$ be the values of the optimal integral solution and optimal fractional solution of an instance $I$, then the \emph{integrality ratio} of $I$ is defined as $\frac{OPT(I)}{OPT_{LP}(I)}$. The \emph{integrality ratio} of the LP is the supremum of the ratio between the value of the optimal integral solution and that of the optimal fractional solution, i.e.\ $\sup_I \frac{OPT(I)}{OPT_{LP}(I)}$. Let $\lvert I \rvert$ be the number of vertices of the instance $I$. The \emph{integrality ratio} of instances with $n$ vertices is defined as $\sup_{\lvert I \rvert=n}\frac{OPT(I)}{OPT_{LP}(I)}$.

The exact integrality ratio of the subtour LP is still unknown. For the \textsc{Metric TSP} the currently best lower and upper bounds are $\frac{4}{3}$ \cite{williamson1990analysis} and $\frac{3}{2}$ \cite{wolsey1980heuristic}, respectively. For the Euclidean case the same upper bound applies and Hougardy gave a lower bound of $\frac{4}{3}$ \cite{Hougardy}. The belief is that the exact integrality ratio is $\frac{4}{3}$, this is also known as the \emph{$\frac{4}{3}$-Conjecture}. The conjecture was proven for instances whose optimal fractional solutions satisfy certain properties \cite{BOYD2011525}. 

Benoit and Boyd computed the exact integrality ratio for \textsc{Metric TSP} instances with a small fixed number of vertices \cite{benoit2008finding}. For $6\leq n \leq 10$ vertices they determined the instances maximizing the integrality ratio with computer assistance and discovered that these instances are unique and have certain structures. Based on these results they generalized the instances to arbitrary number of vertices and conjectured that they maximize the integrality ratio. Later, Boyd and Elliott-Magwood could further decrease the computation time significantly by exploiting more structure of the subtour polytope. The computational results could be extended to $n=12$ \cite{boyd2010structure}. 

Hougardy and Zhong introduced a family of \textsc{Euclidean TSP} instances called the \emph{tetrahedron instances} that has a different structure as the instances from \cite{benoit2008finding} with integrality ratio converging to $\frac{4}{3}$ \cite{Hougardy2020}. They also investigate the runtime of the currently fastest TSP solver \texttt{Concorde} \cite{concorde} to solve the tetrahedron instances in practice. More precisely, they compare the runtime of the tetrahedron instances with the family of instances proposed in \cite{Hougardy} and instances from the TSPLIB, a library of TSP instances \cite{reinelt1991tsplib}. It turned out that the tetrahedron instances are significantly harder to solve than the other instances in practice: \texttt{Corcorde} needs up to 1,000,000 more time to solve the instances compared to TSPLIB instances of similar size.

\noindent\textbf{New results.}
We describe a procedure to construct families of \textsc{Euclidean TSP} instances whose integrality ratios converge to $\frac{4}{3}$. These instances can have a different structure than the known instances from the literature.

We use the same approach as Benoit and Boyd to compute the exact integrality ratio for \textsc{Rectilinear TSP} with $6\leq n\leq 10$ vertices. Using the results of the computations we define the instances $I^2_6, \dots, I^2_{10}$.

\begin{theorem}
The instances $I^2_n$ maximize the integrality ratio for \textsc{Rectilinear TSP} for $n\leq 10$.
\end{theorem}

The instances $I^2_n$ show the same structure as the instances maximizing the integrality ratio in the metric case described in \cite{benoit2008finding}. Based on this we state the following conjecture.

\begin{conjecture} \label{structure conjecture}
The instances maximizing the integrality ratio among all instances with a fixed number of vertices have the following structure: An optimal fractional solution $x^*$ of the subtour LP satisfies $x^*(e)=\frac{1}{2}$ for all edges $e$ of two disjoint triangles and $x^*(e)=0$ or $x^*(e)=1$ for all other edges $e$.
\end{conjecture}

We analyze the structure of $I^2_n$ and generalize the family of instances to arbitrary numbers of vertices. We compute the integrality ratio of the family and show that the integrality ratios of the family converge to $\frac{4}{3}$ as $n\to \infty$. 

Moreover, we also investigate the integrality ratio for \textsc{Metric TSP} and \textsc{Multidimensional Rectilinear TSP}.

\begin{theorem}
Assuming Conjecture \ref{structure conjecture} the \textsc{Metric TSP} instances given in \cite{benoit2008finding} maximize the integrality ratio.
\end{theorem}

For the \textsc{Multidimensional Rectilinear TSP} we define a family of instances $I_n^3$ and show that these instances have the same integrality ratio as the \textsc{Metric TSP} instances given in \cite{benoit2008finding}. Therefore, we conclude:

\begin{theorem}
Assuming Conjecture \ref{structure conjecture} the family of instances $I^3_n$ maximizes the integrality ratio for \textsc{Multidimensional Rectilinear TSP}.
\end{theorem}

We investigate local optima of instances that can be embedded into $\R^d$ such that the distances arise from a totally differentiable norm. Such an instance is a local optimum if we cannot increase the integrality ratio by moving the vertices slightly. A criterion is given to detect local optima. Based on that we give a local search algorithm that computes a local optimum.

For the \textsc{Euclidean TSP} we use the local search algorithm to find instances with high integrality ratio. The results have similar structures as $I^2_n$ in the rectilinear case and the instances from \cite{benoit2008finding} in the metric case. Based on these we give an efficient algorithm generating instances having this structure. Using this algorithm, we were able to generate instances with high integrality ratio for \textsc{Euclidean TSP} and show that their integrality ratios converge to $\frac{4}{3}$. Moreover, in a comparison we see that their integraltiy ratios are higher than these of the \textsc{Euclidean TSP} instances given in \cite{Hougardy} and \cite{Hougardy2020}.

Furthermore, we investigate the runtime of the \texttt{Concorde} TSP solver on slightly modified $I_n^3$. We observe that the runtimes of \texttt{Concorde} for these instances are much higher than for the hard to solve tetrahedron instances given in \cite{Hougardy2020}.

\noindent\textbf{Outline of the paper.}
First, we start with some preliminaries in Section \ref{sec preliminaries}. Section \ref{sec construction of instances} describes a procedure of generating families of instances for \textsc{Euclidean TSP} whose integrality ratios converge to $\frac{4}{3}$. These instances can have a different structure than the currently known families of instances whose integrality ratios converge to $\frac{4}{3}$. In Section \ref{sec computation rect} we compute the exact integrality ratio for \textsc{Rectilinear TSP} instances with a small fixed number of vertices. In the following Section \ref{sec int ratio rectilinear} we generalize the instances maximizing the integrality ratio we found in the previous section to arbitrary numbers of vertices. In Section \ref{sec int ratio metric} and \ref{sec int ratio rectilinear multi} we identify the instances maximizing the integrality ratio assuming Conjecture \ref{structure conjecture} for the \textsc{Metric} and \textsc{Multidimensional Rectilinear TSP}, respectively. 

Then, we give in Section \ref{sec local opt} a criterion that certifies local optimality with respect to integrality ratio. Based on this criterion we describe a local search algorithm to find local optima. In Section \ref{sec int ratio euclidean} we give a more efficient algorithm to generate a family of instances for the \textsc{Euclidean TSP} that was found by the local search algorithm for small number of vertices and has similar structure to the instances maximizing the integrality ratio for other TSP variants. We compare the lower bounds we found for the various TSP variants in Section \ref{sec compare}. Then, we observe in Section \ref{sec runtime concorde} that \texttt{Concorde} needs significantly more running time to solve slightly modified instances we found than to solve the hard to solve instances from the literature.
\section{Preliminaries} \label{sec preliminaries}
\subsection{$p$-Norm}
The $p$-norm $\Vert\cdot\Vert_p$ for $p\geq 1$ in $\R^d$ is defined for all $x=(x_1,\dots, x_d)\in \R^d$ by $\Vert x\Vert_p:=\sqrt[p]{\sum_{i=1}^d\lvert x_i \rvert^p}$, where $\sqrt[1]{z}:=z$. The 1-norm and 2-norm are also called the \emph{Manhatten} and \emph{Euclidean} norm, respectively. For two points $x,y\in \R^d$ we denote the distance of $x$ and $y$ according to the $p$-norm by $\dist_p(x,y)$, i.e.\ $\dist_p(x,y):=\Vert x-y\Vert_p$.
\subsection{\textsc{2-Edge Connected Spanning Subgraph LP}}
Cunningham \cite{monma1990minimum}, Goemans and Bertsimas \cite{goemans1993survivable} showed that in the metric case the optimal solution does not change if we omit the degree constraints, i.e.\ the optimal solutions of the following LP is equal to the optimal solutions of the subtour LP:
\begin{align*}
\min \sum_{e\in E(K_n)} &c(e)x_e\\
\sum_{e \in E(\delta(X))}x_e &\geq 2 && \text{for all~} \emptyset\subset X \subset V(K_n)\\
0 \leq x_e&\leq 1 && \text{for all~} e\in E(K_n)
\end{align*}
This above LP is an LP relaxation for the \textsc{2-Edge Connected Spanning Subgraph} problem where the task is to find a 2-edge connected spanning subgraph of a given graph. 
\begin{theorem}[Cunningham \cite{monma1990minimum}, Goemans and Bertsimas \cite{goemans1993survivable}] \label{2ECSS LP}
If the costs satisfy the triangle inequality, the optimal solutions of the \textsc{2-Edge Connected Spanning Subgraph} LP is the same as that of the subtour LP.
\end{theorem}

\subsection{Structure of Euclidean Tours}
A well-known result about optimal tours for \textsc{Euclidean TSP} is that they do not intersect themself unless all vertices lie on a line.

\begin{lemma}[Flood 1956~\cite{flood1956traveling}] \label{lemma:nocrossing}
Unless all vertices lie on one line, an optimal tour of a \textsc{Euclidean TSP} instance is a simple polygon.
\end{lemma}

An important consequence of Lemma~\ref{lemma:nocrossing} is the following result:

\begin{lemma}[\cite{convexhull}, page 142] \label{lemma:convexhull}
An optimal tour of a \textsc{Euclidean TSP} instance visits the vertices on the boundary of the convex hull 
of all vertices in their cyclic order. 
\end{lemma}

\subsection{Karamata's inequality}

\begin{definition}
A sequence of real numbers $x_1,\dots, x_n$ \emph{majorizes} another sequence $y_1, \dots, y_n$ if
\begin{align*}
x_1\geq x_2 &\dots \geq x_n\\
y_1\geq y_2 &\dots \geq y_n\\
\sum_{i=1}^j x_i&\geq \sum_{i=1}^j y_i \qquad \forall j<n\\
\sum_{i=1}^n x_i&=\sum_{i=1}^n y_i.
\end{align*}
\end{definition}

\begin{theorem}[Karamata's inequality \cite{karamata1932inegalite}] \label{Karamata}
Let $I$ be an interval of real numbers and $f:I\to \R$ be a convex function. Moreover, let $x_1, \dots, x_n$ and $y_1, \dots, y_n$ be sequences of numbers in $I$ such that $(x_1,\dots, x_n)$ majorizes $(y_1,\dots,y_n)$.
Then
\begin{align*}
\sum_{i=1}^n f(x_i)\geq \sum_{i=1}^n f(y_i).
\end{align*}
\end{theorem}

\section{Construction of Instances with Integrality Ratio Converging to $\frac{4}{3}$} \label{sec construction of instances}
In this section we describe a procedure to construct families of instances whose integrality ratios converge to $\frac{4}{3}$. The instances created this way can have a completely different structure than the known instances from the literature.

\subsection{Construction}
We choose a planar embedding of a 2-edge-connected graph $G$ in $\R^2$. For all $k\in \N$ choose a \textsc{Euclidean TSP} instances $G_k$ consisting of the embedded vertices of $G$ and a set of vertices $vw_{1},vw_{2},\dots, vw_{l_k}$ in this order subdividing the line segment $vw$ for every edge $\{v,w\} \in E(G)$. Note that the number of subdividing vertices $l_k$ may differ for different edges of $G$. We call the pairs of vertices of the form $vw_i,vw_{i+1}$ and the pairs $v,vw_1$ and $vw_{l_k},w$ \emph{consecutive vertices}. Let $\delta_k$ be the greatest distance between two consecutive vertices in $G_k$. We further require that the instances $G_k$ satisfy the condition $\lim_{k\to \infty}\delta_k=0$.

Let $T_k^*$ and $x^*_{k}$ be an optimal tour and optimal fractional tour for $G_k$, respectively. Moreover, let $J$ be an optimal $T$-join of the vertices with odd degree in $G$.

\begin{lemma}
For the optimal fractional tour $x^*_k$ of the instance $G_k$ we have $c(x^*_{k})\leq c(E(G))$ for all $k$.
\end{lemma}

\begin{proof}
Set $x(e)=1$ for all edges $e$ connecting two consecutive vertices and $x(e)=0$ for all other edges $e$. This is a solution to the LP relaxation of the \textsc{2-Edge Connected Spanning Subgraph} LP: After the subdivision, the graph stays 2-connected, hence each cut goes through at least two edges and has $x$-value at least 2. The cost of this solution is exactly the cost of $E(G)$. By Theorem \ref{2ECSS LP}, this is also an upper bound for the optimal solution of the subtour LP.
\end{proof}

\begin{lemma}
For the optimal tour $T^*_k$ of the instance $G_k$ we have $\lim_{k\to \infty} c(T^*_{k})\geq c(E(G))+c(J)$.
\end{lemma}

\begin{proof}
For $\epsilon_1>0$ we construct a new instance $G_{k,\epsilon_1}$ from $G_k$ by deleting all subdividing vertices with distance less than $\epsilon_1$ to any vertex $v\in G$. Let $U_{k,\epsilon_1}$ be the set of edges connecting two consecutive vertices of $G_{k}$ where at least one vertex is deleted in $G_{k,\epsilon_1}$. For a vertex $p\in V(G_{k,\epsilon_1})$ and an edge $e\in E(G)$ we say that $p$ \emph{lies} on $e$ if $p\in e$ or $p$ is a vertex subdividing $e$. 

Let $\epsilon_2$ be the shortest distance of two vertices $p,q\in E(G_{k,\epsilon_1})$ that do not lie on a common edge $e\in E(G)$ and $\alpha:=\min_{\{u,w\},\{w,v\} \in E(G) \atop \{u,w\}\neq\{w,v\}} \angle uwv$ be the smallest angle between two different edges with a common vertex in $E(G)$. We claim that for $\epsilon_1$ fixed we have $\liminf_{k\to \infty}\epsilon_2>0$. Since the embedding is planar, we have $\alpha>0$ and $\liminf_{k\to \infty}\dist_2(p,q)>0$ if $p$ and $q$ lie on different edges of $G$ not incident to each other. Now, let $p$ and $q$ lie on two different edges $\{u,w\},\{w,v\}\in E(G)$ with a common vertex, respectively. If $\angle uvw >\frac{\pi}{2}$, then $\dist_2(p,q)\geq \epsilon_1$. Else, for $p$ fixed the distance $\dist_2(p,q)$ is minimized if $pq$ is perpendicular to $wv$. Thus, we have $\dist_2(p,q)\geq \dist_2(p,w)\sin(\angle uwv)\geq \epsilon_1\sin(\alpha)>0$ which proves the claim. 

Now, consider the subset of edges $S_{k,\epsilon_1}$ of an optimal tour $T^*_{k,\epsilon_1}$ of $G_{k, \epsilon_1}$ consisting of edges not connecting two vertices lying on the same edge. Recall that by definition every edge in $S_{k,\epsilon_1}$ has length at least $\epsilon_2>0$. Now, if $\lvert S_{k,\epsilon_1} \rvert>\frac{c(E(G))+c(J)}{\epsilon_2}$, we have $c(T^*_k)\geq c(T^*_{k,\epsilon_1})>c(E(G))+c(J)$ since the vertices of $G_{k,\epsilon_1}$ is a subset of that of $G_{k}$. It remains the case that $\lvert S_{k,\epsilon_1} \rvert \leq\frac{c(E(G))+c(J)}{\epsilon_2}$. Note that by Lemma \ref{lemma:nocrossing} the optimal tour is a simple polygon. Thus, the edges in $T^*_{k,\epsilon_1} \backslash S_{k,\epsilon_1}$ are connecting consecutive vertices and we have
\begin{align*}
c(T^*_{k,\epsilon_1}\backslash S_{k,\epsilon_1})&=c(E(G_{k,\epsilon_1})\cap T_{k,\epsilon_1})= c(E(G_{k,\epsilon_1}))-c(E(G_{k,\epsilon_1}) \backslash T_{k,\epsilon_1})\\
&\geq c(E(G_{k,\epsilon_1}))-\lvert S_{k,\epsilon_1} \rvert \delta_k= c(E(G))-\lvert S_{k,\epsilon_1} \rvert \delta_k-c(U_{k,\epsilon_1}) \\
&\geq c(E(G))-\lvert S_{k,\epsilon_1} \rvert\delta_k -2\lvert E(G) \rvert (\epsilon_1+\delta_k)
\end{align*}
since $c(U_{k,\epsilon_1})\leq 2\lvert E(G)\rvert (\epsilon_1+\delta_k)$. Furthermore, the edges in $S_{k,\epsilon_1}$ are a $T$-join for the vertices with odd degree in $(V(G_{k,\epsilon_1}),T^*_{k,\epsilon_1} \backslash S_{k,\epsilon_1})$. Therefore, $S_{k,\epsilon_1} \cup \left( E(G_{k,\epsilon_1}) \backslash T_{k,\epsilon_1} \right)\cup U_{k,\epsilon_1}$ is a $T$-join for the vertices with odd degree in $G_k$ as $E(G_k)=(T^*_{k,\epsilon_1}\backslash S_{k,\epsilon_1}) \cup \left( E(G_{k,\epsilon_1}) \backslash T_{k,\epsilon_1} \right) \cup U_{k,\epsilon_1}$. Since $G$ has the same set of vertices with odd degree as $G_k$, this is also a $T$-join for the vertices with odd degrees in $G$. Thus, 
\begin{align*}
c(S_{k,\epsilon_1})\geq c(J)-c(E(G_{k,\epsilon_1}) \backslash T_{k,\epsilon_1})-c(U_{k,\epsilon_1}) \geq c(J)-\lvert S_{k,\epsilon_1} \rvert \delta_k-2\lvert E(G) \rvert (\epsilon_1+\delta_k). 
\end{align*}
Altogether, for the total length of the tour we have:
\begin{align*}
\lim_{k\to\infty}c(T^*_k)\geq &\lim_{\epsilon_1\to 0}\lim_{k\to \infty}c(T^*_{k,\epsilon_1})=\lim_{\epsilon_1\to 0}\lim_{k\to \infty}c(T^*_{k,\epsilon_1}\backslash S_{k,\epsilon_1})+c(S_{k,\epsilon_1})\\
\geq &\lim_{\epsilon_1 \to 0}\lim_{k\to \infty} c(E(G))-\lvert S_{k,\epsilon_1} \rvert\delta_k-2\lvert E(G) \rvert(\epsilon_1+\delta_k)\\
&+c(J)-\lvert S_{k,\epsilon_1} \rvert\delta_k-2\lvert E(G) \rvert (\epsilon_1+\delta_k)\\
\geq &\lim_{\epsilon_1 \to 0}\lim_{k \to \infty} c(E(G))-2\frac{c(E(G))+c(J)}{\epsilon_2}\delta_k-4\lvert E(G) \rvert(\epsilon_1+\delta_k)+c(J)\\
=&\lim_{\epsilon_1\to 0}c(E(G))-4\lvert E(G) \rvert\epsilon_1+c(J)=c(E(G))+c(J).\qedhere
\end{align*}
\end{proof}

Hence, we conclude:

\begin{theorem} \label{int ratio convergence}
The integrality ratios of the family of instances $G_k$ converge as $k\to \infty$ to at least $\frac{c(E(G))+c(J)}{c(E(G))}$, where $J$ is the cost of the optimal $T$-join of the vertices with odd degree in $G$.
\end{theorem}

\begin{remark}
It is not possible to construct instances with higher integrality ratio than $\frac{4}{3}$ using this procedure. Let a planar embedding of a 2-connected graph $G$ be given. Consider the vector $y\in \R^{E(G)}$ with $y(e)=\frac{1}{3}$ for every $e\in E(G)$. Since the graph is 2-edge-connected, every cut $S$ that intersets an odd number of edges intersects $E(G)$ at least three times. Hence, for every odd cut $S$ the total $y$-value of $\delta(S)$ is at least $3\cdot\frac{1}{3}=1$. Therefore, $y$ lies in the $T$-join polytope and has cost $\frac{c(E(G))}{3}$. Thus, we have $c(J)\leq \frac{c(E(G))}{3}$. A similar construction of a vector in the $T$-join polytope was already used in \cite{monma1990minimum}.
\end{remark}

In the next section we will see concrete examples of $G_k$ whose integrality ratios converge to $\frac{4}{3}$.

\subsection{Applications}
Now, we apply the results of the last section to construct families of instances whose integrality ratios converge to $\frac{4}{3}$.

The \emph{tetrahedron instances} were already introduced in \cite{Hougardy2020}. It consists of the vertices $A,B,C$ forming an equilateral triangle with the center $M$. The sides of the triangle $AB,BC,CA$ and the segments $MA,MB,MC$ are subdivided equidistantly by $a_k$ and $b_k$ equidistant vertices, respectively (Figure \ref{tetrahedron}). Moreover, we have $a_k,b_k\to \infty$ as $k\to \infty$. We can apply Theorem \ref{int ratio convergence} to get another proof of Theorem 3.19 in \cite{Hougardy2020}, that the integrality ratio of this family converges to $\frac{4}{3}$ as $k\to \infty$: We take a complete graph $K_4$ and embed it to the Euclidean plane such that the vertices coincide with the vertices $A,B,C,M$ of the instance. Now, every vertex has an odd degree in $G$. Hence, a $T$-join has to correct the parity of every vertex and the cost of the $T$-join is at least $\dist(A,B)+\dist(C,M)$ which is by symmetry $\frac{c(E(G))}{3}$.

\begin{figure}
\centering
\begin{tikzpicture}[scale=1.35]
\def\n{7}
\def\m{6}
\def\sqrt3{1.73205}
\def\mycircle{circle(0.5mm)}

\pgfmathparse{\n - 1}
\foreach \i in {1,...,\pgfmathresult}
{
   \fill (\n     - \i / 2, \i * \sqrt3 / 2) \mycircle;
   \pgfmathparse{(\n - \i) * \sqrt3 / 2}
   \fill (\n / 2 - \i / 2,  \pgfmathresult) \mycircle;
   \fill (             \i,               0) \mycircle;
}

\pgfmathparse{\m - 1}
\foreach \j in {1,...,\pgfmathresult}
{
   \fill(     \j * \n / 2 / \m, \j * \n / 2 / \sqrt3 / \m) \mycircle;
   \fill(\n - \j * \n / 2 / \m, \j * \n / 2 / \sqrt3 / \m) \mycircle;
   \fill(               \n / 2, \n * \sqrt3 / 2 - \j * \n / \sqrt3 / \m) \mycircle;
}

\fill(\n / 2 ,  \n / 2 / \sqrt3) \mycircle node[anchor = south] {$M$};
\fill(\n, 0) \mycircle node[anchor = west] {$B$};
\fill(\n / 2,  \n * \sqrt3 / 2) \mycircle node[anchor = south] {$C$};
\fill(0, 0) \mycircle node[anchor = east] {$A$};
\end{tikzpicture}
\caption{The tetrahedron instance: The graph $G$ is a complete graph consisting of the vertices $A,B,C,M$. In the figure the edges are subdivided equidistantly by $a_k=6$ and $b_k=5$ vertices, respectively.}
\label{tetrahedron}
\end{figure}
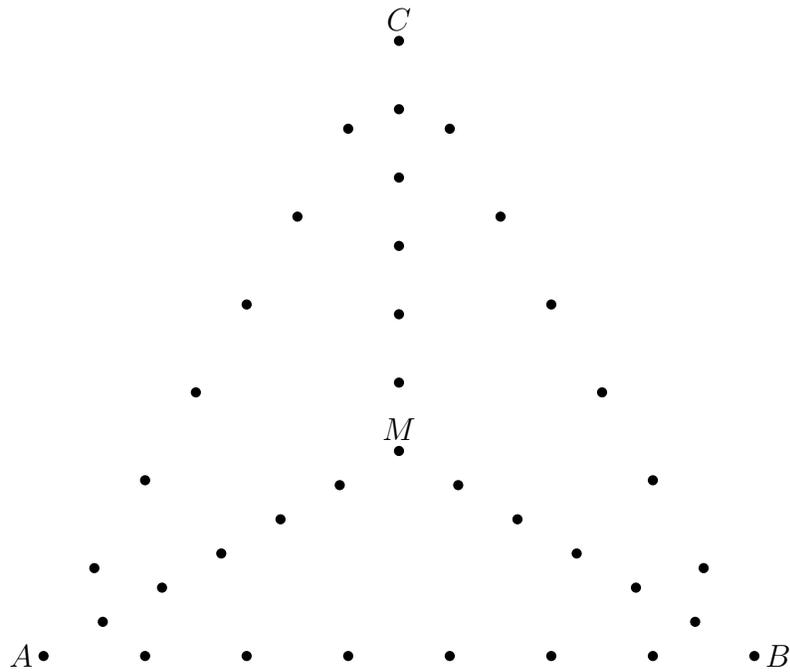

Another example are the \emph{hexagon instances}. In contrast to the tetrahedron instances we do not uniquely define the graph $G$ we subdivide. We only require that the vertices of $G$ form small regular hexagons that tesselate a subset of $\R^2$ (Figure \ref{hexagon}). Every vertex of $G$ not lying on the border of the tesselation has odd degree. Thus, every $T$-join is incident to each of these vertices. The cost of each path in the $T$-join can be bounded by the shortest distance between two distinct vertices, which is a side length of the hexagons. Hence, if we tesselate a subset of $\R^2$ where the number of vertices on the border is small compared to the total number of vertices, the cost of the $T$-join is at least roughly $\frac{c(E(G))}{3}$. We can take growing tesselations and subdivide the edges to get instances with integrality ratio converging to $\frac{4}{3}$.
\begin{figure}
        \centering
        \input{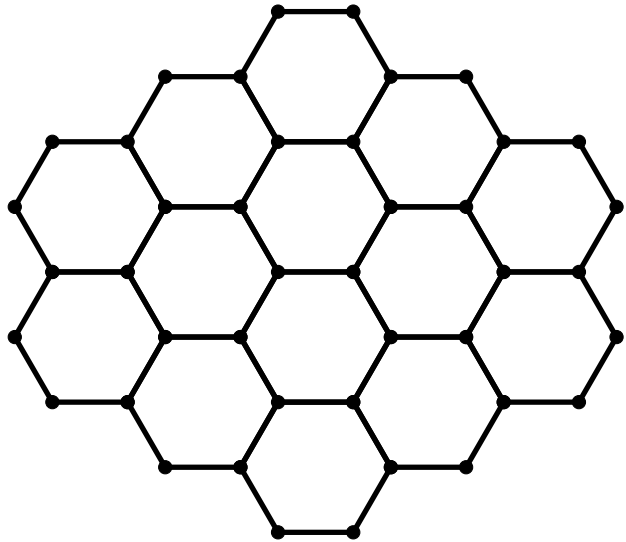}
  \caption{A possible graph $G$ for the hexagon instances. One possible construction of the instances $G_k$ is to subdivide every edge by $k$ equidistant vertices.}
  \label{hexagon}
\end{figure}

\section{Computing the Exact Integrality Ratio for Rectilinear TSP} \label{sec computation rect}
In this section we compute the exact integrality ratio for \textsc{Rectilinear TSP} instances with a small fixed number of vertices. For that we used Sylvia Boyd's list of the extremal points of the subtour polytope published on her homepage \cite{extremePoints}. This approach was already used to compute the exact integrality ratio of \textsc{Metric TSP} instances with $6\leq n\leq 12$ \cite{benoit2008finding,boyd2010structure}. 

Let $\mathcal{T}$ be the set of all tours for the vertex set $V(K_n)$. Given an extremal point $x$ of the subtour polytope we solve the corresponding program given by
\begin{align}
\max\ &f\\
\sum_{\{u,v\} \in T} \lvert u_x-v_x \rvert + \lvert u_y-v_y \rvert &\geq f && \text{for all~} T\in \mathcal{T} \label{cond tour length} \\
\sum_{u,v \in K_n} \left(\lvert u_x-v_x \rvert + \lvert u_y-v_y \rvert  \right) \cdot x(\{u,v\}) &\leq 1 \label{cond frac length}\\
u_x,u_y,f&\in \R && \text{for all~} u\in V(K_n)
\end{align}
The variables $u_x$ and $u_y$ for all $u\in V(K_n)$ represent the $x$- and $y$-coordinate of the vertex $u$ and $f$ represents the integrality ratio of the instance.

On the one hand, each of the programs gives a lower bound for the integrality ratio since condition (\ref{cond tour length}) ensures that the length of the optimal tour is at least $f$ and condition (\ref{cond frac length}) ensures that the cost of the optimal fractional tour is at most 1. Therefore, their ratio is at least $f$. On the other hand, the maximum value of the programs for all extremal points $x$ of the subtour LP is the maximum achievable integrality ratio, since the instance achieving the maximal integrality ratio can be scaled such that the optimal fractional tour has length 1. Thus, the largest $f$ of all programs corresponding to all extremal points with a fixed number of vertices is the exact integrality ratio.

For performance reasons we added the constraints $0 \leq f\leq 2$ and $0\leq u_x,u_y\leq 1$. This can be done without loss of generality since by Wolsey's analysis $f$ lies between 1 and $\frac{3}{2}$ \cite{wolsey1980heuristic}. Moreover, every instance with a fractional tour of cost at most 1 fits into the unit square since the fractional tour intersects every cut at least twice. Therefore, by translating this instance we can assume that it lies in the unit square.

Note that it is known that for $n\leq 5$ every optimal fractional tour is also integral and thus the integrality ratio is 1. For $6 \leq n \leq 10$ vertices and every extremal point $x$ of the subtour polytope the corresponding program has been solved using \texttt{Gurobi 8.11}. It took about a month of computation time to solve all 461 programs for the case $n=10$ with an \texttt{Intel i5-4670}. On the same machine a month of computation time was not sufficient to solve a single program for $n=11$.
 
The resulting instances are shown in Figure \ref{results rectilinear}. They have a similar structure as the instances maximizing the integrality ratio in the metric case. Their explicit coordinates will be given in Subsection \ref{sub sec instance Iijk} where we generalize these instances to higher $n$. 

\begin{figure}[!htb]
\begin{minipage}[t]{0.32\textwidth}
	\centering
\begin{tikzpicture}[scale=0.3]
    \node[circle,fill=black, inner sep=2pt] at (0, 0) {};
    \node[circle,fill=black, inner sep=2pt] at (0, 10) {};
    \node[circle,fill=black, inner sep=2pt] at (7.5, 0) {};
    \node[circle,fill=black, inner sep=2pt] at (7.5, 10) {};
    \node[circle,fill=black, inner sep=2pt] at (2.5, 5) {};
    \node[circle,fill=black, inner sep=2pt] at (5, 5) {};
\end{tikzpicture}
	\caption*{$n=6$, ratio=$\frac{18}{17}$}
	\end{minipage}
	\hfill%
	\begin{minipage}[t]{0.32\textwidth}
	\centering	
\begin{tikzpicture}[scale=0.3]
    \node[circle,fill=black, inner sep=2pt] at (5, 5) {};
    \node[circle,fill=black, inner sep=2pt] at (6.67, 10) {};
    \node[circle,fill=black, inner sep=2pt] at (6.67, 0) {};
    \node[circle,fill=black, inner sep=2pt] at (1.67, 5) {};
    \node[circle,fill=black, inner sep=2pt] at (0, 10) {};
    \node[circle,fill=black, inner sep=2pt] at (0, 0) {};
    \node[circle,fill=black, inner sep=2pt] at (3.33, 5) {};
\end{tikzpicture}
	\caption*{$n=7$, ratio=$\frac{13}{12}$}
	\end{minipage}
  \hfill%
	\begin{minipage}[t]{0.32\textwidth}
	\centering	
\begin{tikzpicture}[scale=0.3]
    \node[circle,fill=black, inner sep=2pt] at (3.75, 5) {};
    \node[circle,fill=black, inner sep=2pt] at (1.25, 5) {};
    \node[circle,fill=black, inner sep=2pt] at (6.25, 10) {};
    \node[circle,fill=black, inner sep=2pt] at (6.25, 0) {};
    \node[circle,fill=black, inner sep=2pt] at (0, 0) {};
    \node[circle,fill=black, inner sep=2pt] at (0, 10) {};
    \node[circle,fill=black, inner sep=2pt] at (5, 5) {};
    \node[circle,fill=black, inner sep=2pt] at (2.5, 5) {};
\end{tikzpicture}
	\caption*{$n=8$, ratio=$\frac{34}{31}$}
	\end{minipage}
\hfill%
\vspace{5mm}
\begin{minipage}[t]{0.32\textwidth}
	\centering
\begin{tikzpicture}[scale=0.3]
    \node[circle,fill=black, inner sep=2pt] at (3.08, 6.15) {};
    \node[circle,fill=black, inner sep=2pt] at (4.26, 6.15) {};
    \node[circle,fill=black, inner sep=2pt] at (3.85, 10) {};
    \node[circle,fill=black, inner sep=2pt] at (7.69, 0) {};
    \node[circle,fill=black, inner sep=2pt] at (0, 0) {};
    \node[circle,fill=black, inner sep=2pt] at (1.54, 6.15) {};
    \node[circle,fill=black, inner sep=2pt] at (0, 10) {};
    \node[circle,fill=black, inner sep=2pt] at (7.69, 10) {};
    \node[circle,fill=black, inner sep=2pt] at (6.15, 6.15) {};
\end{tikzpicture}
	\caption*{$n=9$, ratio=$\frac{31}{28}$}
	\end{minipage}
	  \hfill%
	\begin{minipage}[t]{0.32\textwidth}
	\centering	
	\begin{tikzpicture}[scale=0.3]
    \node[circle,fill=black, inner sep=2pt] at (6, 5) {};
    \node[circle,fill=black, inner sep=2pt] at (4, 5) {};
    \node[circle,fill=black, inner sep=2pt] at (5, 10) {};
    \node[circle,fill=black, inner sep=2pt] at (0, 0) {};
    \node[circle,fill=black, inner sep=2pt] at (10, 0) {};
    \node[circle,fill=black, inner sep=2pt] at (8, 5) {};
    \node[circle,fill=black, inner sep=2pt] at (10, 10) {};
    \node[circle,fill=black, inner sep=2pt] at (0, 10) {};
    \node[circle,fill=black, inner sep=2pt] at (2, 5) {};
    \node[circle,fill=black, inner sep=2pt] at (5, 0) {};
\end{tikzpicture}
	\caption*{$n=10$, ratio=$\frac{28}{25}$}
	\end{minipage}
	\caption{Instances with $n$ vertices maximizing the integrality ratio for \textsc{Rectilinear TSP}.}
	\label{results rectilinear}
\end{figure}

We see that vertices are not distributed equally on the three lines as in the metric case. Here we get a higher integrality ratio if the number of vertices is higher on the center line. In the following sections we will investigate this phenomenon and other structural properties of this family of instances.

\section{Integrality Ratio for Rectilinear TSP} \label{sec int ratio rectilinear}
In this section we construct a family of \textsc{Rectilinear TSP} instances with similar structure and properties as the results of the computations from Section \ref{sec computation rect} and analyze their integrality ratio.

\subsection{Structure of the Fractional Tours}
The optimal fractional tours of the instances maximizing the integrality ratio found in Section~\ref{sec computation rect} have the same form. In this section we describe and extend it to higher number of vertices.

Note that every fractional tour $x$ can be interpreted as a weighted complete graph where the vertex set consists of the vertices of the instance and the weight of the edge $e$ is equal to $x(e)$. 
We define the weighted graphs $x_{i,j,k}$ for nonnegative integers $i,j,k$ as follows:
The vertex set $V(x_{i,j,k}):=\{X_0,\dots, X_{i+1},Y_0, \dots, Y_{j+1}, Z_0, \dots,  Z_{k+1}\}$ consists of $i+j+k+6$ vertices. We set the weight of the edges $\{X_r,X_{r+1}\}, \{Y_s,Y_{s+1}\}, \{Z_t,Z_{t+1}\}$ for all $r\in \{0,\dots,  i\}, s\in \{0,\dots, j\}, t\in \{0,\dots, k\}$ to $1$. Moreover, we set the weight of the edges $\{X_0,Y_0\}, \{X_0,Z_0\}$, $\{Y_0,Z_0\}$, $\{X_{i+1},Y_{j+1}\}, \{X_{i+1},Z_{k+1}\}$ and $\{Y_{j+1},Z_{k+1}\}$ to $\frac{1}{2}$. All other edges have weight 0 (Figure \ref{Gijk}, the instance $I_{2,2,1}$ will be defined later).

 The optimal fractional tours of the instances maximizing the integrality ratio found in Section \ref{sec computation rect} are isomorphic to $x_{i,j,k}$ for some $i,j,k$.

\begin{figure}[ht]
\centering
\begin{tikzpicture}[line cap=round,line join=round,>=triangle 45,x=0.4cm,y=0.4cm]
\draw [line width=2pt] (80,0)-- (84,10);
\draw [line width=2pt] (84,10)-- (80,20);
\draw [line width=2pt] (80,20)-- (80,0);
\draw [line width=2pt] (100,0)-- (96,10);
\draw [line width=2pt] (96,10)-- (100,20);
\draw [line width=2pt] (100,20)-- (100,0);
\draw [line width=2pt] (100,0)-- (90,0);
\draw [line width=2pt] (90,0)-- (80,0);
\draw [line width=2pt] (84,10)-- (88,10);
\draw [line width=2pt] (88,10)-- (92,10);
\draw [line width=2pt] (92,10)-- (96,10);
\draw [line width=2pt] (80,20)-- (90,20);
\draw [line width=2pt] (90,20)-- (100,20);
\begin{scriptsize}
\draw [fill=black] (92,10) circle (2.5pt);
\draw[color=black] (92.31864641128108,10.865879379548424) node {$Y_2$};
\draw [fill=black] (88,10) circle (2.5pt);
\draw[color=black] (88.30650468370374,10.865879379548424) node {$Y_0$};
\draw [fill=black] (90,20) circle (2.5pt);
\draw[color=black] (90.31257554749241,20.856112281215985) node {$Z_1$};
\draw [fill=black] (80,0) circle (2.5pt);
\draw[color=black] (80.92234264582482,0.8756464778808611) node {$X_0$};
\draw [fill=black] (100,0) circle (2.5pt);
\draw[color=black] (100.70280844916001,0.8756464778808611) node {$X_2$};
\draw [fill=black] (96,10) circle (2.5pt);
\draw[color=black] (95.83078813885844,10.865879379548424) node {$Y_3$};
\draw [fill=black] (100,20) circle (2.5pt);
\draw[color=black] (100.30280844916001,20.856112281215985) node {$Z_2$};
\draw [fill=black] (80,20) circle (2.5pt);
\draw[color=black] (80.32234264582482,20.856112281215985) node {$Z_0$};
\draw [fill=black] (84,10) circle (2.5pt);
\draw[color=black] (84.33448437340216,10.865879379548424) node {$Y_1$};
\draw [fill=black] (90,0) circle (2.5pt);
\draw[color=black] (90.31257554749241,0.8756464778808611) node {$X_1$};
\end{scriptsize}
\end{tikzpicture}
  \caption{The instance $I_{2,2,1}$ with optimal fractional tour $x_{2,2,1}$. The straight and dashed edges have weights 1 and $\frac{1}{2}$ in $x_{2,2,1}$, respectively.}
  \label{Gijk}
\end{figure}
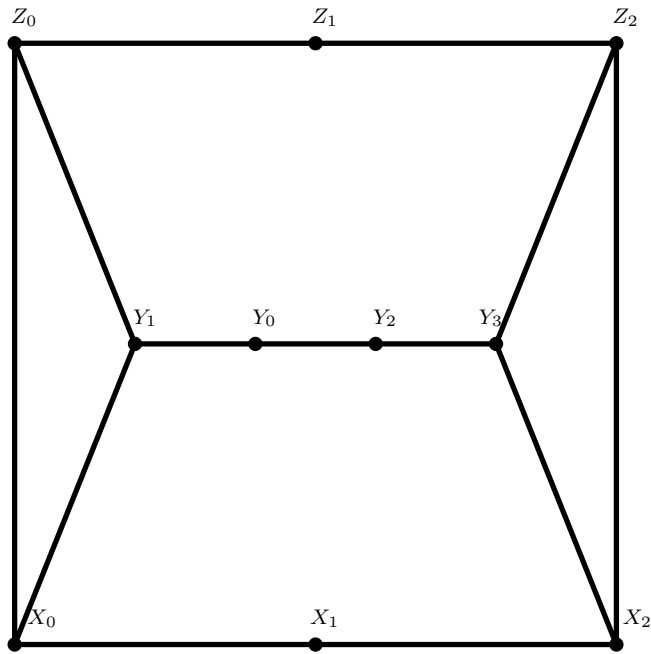

\subsection{Structure of the Optimal Tours}
In this subsection we describe the structure of the optimal tours of the instances maximizing the integrality ratio computed in Section \ref{sec computation rect} which is the motivation for the construction of the generalized instances $I_{i,j,k}^2$ in the next subsection.

A \emph{pseudo-tour} of a TSP instance is a closed walk that visits every vertex at least once. We first define a set of pseudo-tours $\mathfrak{T}$ which is the union of three sets of pseudo-tours $T^\uparrow, T^\circ, T^\downarrow$ and the pseudo-tours $T^\nwarrow, T^\nearrow, T^\leftarrow, T^\rightarrow, T^\swarrow, T^\searrow$.

Let $T^\uparrow:=\{T_0^\uparrow,\dots, T_k^\uparrow\}$ be a set of pseudo-tours where the pseudo-tour $T_l^\uparrow$ for some $0\leq l \leq k$ consists of (Figure \ref{TupM}):
\begin{itemize}
\item two copies of the edges $\{Z_s,Z_{s+1}\}$ for all $0\leq s \leq k, s\neq l$
\item a copy of the edges $\{Y_s,Y_{s+1}\}$ for all $0\leq s \leq j$
\item a copy of the edges $\{X_s,X_{s+1}\}$ for all $0\leq s \leq i$
\item a copy of the edges $\{Z_0,Y_0\}$, $\{Z_0,X_0\}$, $\{Z_{k+1},Y_{j+1}\}$ and $\{Z_{k+1},X_{i+1}\}$
\end{itemize}

\begin{figure}[!htb]
\begin{minipage}[t]{0.45\textwidth}
	\centering
	\begin{tikzpicture}[line cap=round,line join=round,>=triangle 45,x=0.3cm,y=0.3cm]
\draw [line width=2pt] (84,10)-- (80,20);
\draw [line width=2pt] (80,20)-- (80,0);
\draw [line width=2pt] (100,20.080242834551548)-- (100,0);
\draw [line width=2pt] (100,0)-- (90,0);
\draw [line width=2pt] (90,0)-- (80,0);
\draw [line width=2pt] (84,10)-- (88,10);
\draw [line width=2pt] (88,10)-- (92,10);
\draw [line width=2pt] (92,10)-- (96,10);
\draw [line width=2pt] (90,20.080242834551548)-- (100,20.080242834551548);
\draw [line width=2pt] (100,20.080242834551548)-- (96,10);
\draw [shift={(95,10)},line width=2pt]  plot[domain=1.1103379472689017:2.0312547063208917,variable=\t]({1*11.252168484497902*cos(\t r)+0*11.252168484497902*sin(\t r)},{0*11.252168484497902*cos(\t r)+1*11.252168484497902*sin(\t r)});
\begin{scriptsize}
\draw [fill=black] (92,10) circle (2.5pt);
\draw[color=black] (92.78601426973307,8.733752621016578) node {$Y_2$};
\draw [fill=black] (88,10) circle (2.5pt);
\draw[color=black] (88.76678131491627,8.733752621016578) node {$Y_1$};
\draw [fill=black] (90,20.080242834551548) circle (2.5pt);
\draw[color=black] (89.6248198108884,21.423690377235822) node {$Z_1$};
\draw [fill=black] (80,0) circle (2.5pt);
\draw[color=black] (81.04443485116717,1.1468859197894115) node {$X_0$};
\draw [fill=black] (100,0) circle (2.5pt);
\draw[color=black] (101.05027978356983,1.1468859197894115) node {$X_2$};
\draw [fill=black] (96,10) circle (2.5pt);
\draw[color=black] (96.71492738286858,8.733752621016578) node {$Y_3$};
\draw [fill=black] (100,20.080242834551548) circle (2.5pt);
\draw[color=black] (100.50836073348218,21.423690377235822) node {$Z_2$};
\draw [fill=black] (80,20) circle (2.5pt);
\draw[color=black] (80.59283564276079,21.423690377235822) node {$Z_0$};
\draw [fill=black] (84,10) circle (2.5pt);
\draw[color=black] (84.70238843925885,8.733752621016578) node {$Y_0$};
\draw [fill=black] (90,0) circle (2.5pt);
\draw[color=black] (90.48285830686052,1.1468859197894115) node {$X_1$};
\end{scriptsize}
\end{tikzpicture}
	\caption{The pseudo-tour $T^\uparrow_0$.} 
	\label{TupM}
	\end{minipage}
	\hfill%
	\begin{minipage}[t]{0.45\textwidth}
	\centering	
	\begin{tikzpicture}[line cap=round,line join=round,>=triangle 45,x=0.3cm,y=0.3cm]
\draw [line width=2pt] (84,10)-- (80,20);
\draw [line width=2pt] (100,0)-- (90,0);
\draw [line width=2pt] (90,0)-- (80,0);
\draw [line width=2pt] (88,10)-- (92,10);
\draw [line width=2pt] (92,10)-- (96,10);
\draw [line width=2pt] (100,20.080242834551548)-- (96,10);
\draw [line width=2pt] (80,20)-- (100,20.080242834551548);
\draw [line width=2pt] (80,0)-- (84,10);
\draw [line width=2pt] (96,10)-- (100,0);
\draw [shift={(90,7.8982940854647765)},line width=2pt]  plot[domain=0.810189086976861:2.3314035666129325,variable=\t]({1*2.9012355559644485*cos(\t r)+0*2.9012355559644485*sin(\t r)},{0*2.9012355559644485*cos(\t r)+1*2.9012355559644485*sin(\t r)});
\draw [shift={(94,7.8982940854647765)},line width=2pt]  plot[domain=0.810189086976861:2.3314035666129325,variable=\t]({1*2.9012355559644485*cos(\t r)+0*2.9012355559644485*sin(\t r)},{0*2.9012355559644485*cos(\t r)+1*2.9012355559644485*sin(\t r)});
\begin{scriptsize}
\draw [fill=black] (92,10) circle (2.5pt);
\draw[color=black] (92.65053450721116,8.733752621016578) node {$Y_2$};
\draw [fill=black] (88,10) circle (2.5pt);
\draw[color=black] (88.76678131491627,8.733752621016578) node {$Y_1$};
\draw [fill=black] (90,20.080242834551548) circle (2.5pt);
\draw[color=black] (89.6248198108884,21.423690377235822) node {$Z_1$};
\draw [fill=black] (80,0) circle (2.5pt);
\draw[color=black] (81.24443485116715,1.1468859197894115) node {$X_0$};
\draw [fill=black] (100,0) circle (2.5pt);
\draw[color=black] (101.05027978356983,1.1468859197894115) node {$X_2$};
\draw [fill=black] (96,10) circle (2.5pt);
\draw[color=black] (97.26008730370921,8.733752621016578) node {$Y_3$};
\draw [fill=black] (100,20.080242834551548) circle (2.5pt);
\draw[color=black] (100.50836073348218,21.423690377235822) node {$Z_2$};
\draw [fill=black] (80,20) circle (2.5pt);
\draw[color=black] (80.59283564276078,21.423690377235822) node {$Z_0$};
\draw [fill=black] (84,10) circle (2.5pt);
\draw[color=black] (84.70238843925885,8.733752621016578) node {$Y_0$};
\draw [fill=black] (90,0) circle (2.5pt);
\draw[color=black] (90.48285830686052,1.1468859197894115) node {$X_1$};
\end{scriptsize}
\end{tikzpicture}
	\caption{The pseudo-tour $T^\circ_0$.}
	\label{TmidM}	
	\end{minipage}
\end{figure}
\begin{figure}[!htb]
\begin{minipage}[t]{0.45\textwidth}
	\centering
	\begin{tikzpicture}[line cap=round,line join=round,>=triangle 45,x=0.3cm,y=0.3cm]
\draw [line width=2pt] (84,10)-- (80,20);
\draw [line width=2pt] (100,0)-- (90,0);
\draw [line width=2pt] (90,0)-- (80,0);
\draw [line width=2pt] (88,10)-- (92,10);
\draw [line width=2pt] (92,10)-- (96,10);
\draw [line width=2pt] (80,20)-- (100,20.080242834551548);
\draw [line width=2pt] (96,10)-- (100,0);
\draw [line width=2pt] (80,20)-- (80,0);
\draw [line width=2pt] (84,10)-- (88,10);
\draw [shift={(95,10)},line width=2pt]  plot[domain=1.1103379472689017:2.0312547063208917,variable=\t]({1*11.252168484497902*cos(\t r)+0*11.252168484497902*sin(\t r)},{0*11.252168484497902*cos(\t r)+1*11.252168484497902*sin(\t r)});
\draw [shift={(85.08056478017639,10)},line width=2pt]  plot[domain=1.1167725497075345:2.040868326354815,variable=\t]({1*11.216601021948174*cos(\t r)+0*11.216601021948174*sin(\t r)},{0*11.216601021948174*cos(\t r)+1*11.216601021948174*sin(\t r)});
\begin{scriptsize}
\draw [fill=black] (92,10) circle (2.5pt);
\draw[color=black] (92.78601426973306,8.733752621016578) node {$Y_2$};
\draw [fill=black] (88,10) circle (2.5pt);
\draw[color=black] (88.76678131491627,8.733752621016578) node {$Y_1$};
\draw [fill=black] (90,20.080242834551548) circle (2.5pt);
\draw[color=black] (89.6248198108884,21.423690377235822) node {$Z_1$};
\draw [fill=black] (80,0) circle (2.5pt);
\draw[color=black] (81.04443485116715,1.1468859197894115) node {$X_0$};
\draw [fill=black] (100,0) circle (2.5pt);
\draw[color=black] (101.05027978356985,1.1468859197894115) node {$X_2$};
\draw [fill=black] (96,10) circle (2.5pt);
\draw[color=black] (97.26008730370923,8.733752621016578) node {$Y_3$};
\draw [fill=black] (100,20.080242834551548) circle (2.5pt);
\draw[color=black] (100.50836073348219,21.423690377235822) node {$Z_2$};
\draw [fill=black] (80,20) circle (2.5pt);
\draw[color=black] (80.59283564276078,21.423690377235822) node {$Z_0$};
\draw [fill=black] (84,10) circle (2.5pt);
\draw[color=black] (84.70238843925884,8.733752621016578) node {$Y_0$};
\draw [fill=black] (90,0) circle (2.5pt);
\draw[color=black] (90.48285830686052,1.1468859197894115) node {$X_1$};
\end{scriptsize}
\end{tikzpicture}
	\caption{The pseudo-tour $T^\nwarrow$.} 
	\label{TnwM}
	\end{minipage}
  \hfill%
	\begin{minipage}[t]{0.45\textwidth}
	\centering
	\begin{tikzpicture}[line cap=round,line join=round,>=triangle 45,x=0.3cm,y=0.3cm]
\draw [line width=2pt] (84,10)-- (80,20);
\draw [line width=2pt] (100,0)-- (90,0);
\draw [line width=2pt] (90,0)-- (80,0);
\draw [line width=2pt] (88,10)-- (92,10);
\draw [line width=2pt] (92,10)-- (96,10);
\draw [line width=2pt] (80,20)-- (100,20.080242834551548);
\draw [line width=2pt] (80,0)-- (84,10);
\draw [shift={(90,6.949935747811381)},line width=2pt]  plot[domain=0.9904077994261431:2.1511848541636502,variable=\t]({1*3.6473129756683784*cos(\t r)+0*3.6473129756683784*sin(\t r)},{0*3.6473129756683784*cos(\t r)+1*3.6473129756683784*sin(\t r)});
\draw [shift={(94,6.949935747811381)},line width=2pt]  plot[domain=0.9904077994261431:2.1511848541636502,variable=\t]({1*3.6473129756683784*cos(\t r)+0*3.6473129756683784*sin(\t r)},{0*3.6473129756683784*cos(\t r)+1*3.6473129756683784*sin(\t r)});
\draw [line width=2pt] (84,10)-- (88,10);
\draw [shift={(86,6.949935747811381)},line width=2pt]  plot[domain=0.9904077994261431:2.1511848541636502,variable=\t]({1*3.6473129756683784*cos(\t r)+0*3.6473129756683784*sin(\t r)},{0*3.6473129756683784*cos(\t r)+1*3.6473129756683784*sin(\t r)});
\draw [line width=2pt] (100,20.080242834551548)-- (100,0);
\begin{scriptsize}
\draw [fill=black] (92,10) circle (2.5pt);
\draw[color=black] (92.78601426973307,8.733752621016578) node {$Y_2$};
\draw [fill=black] (88,10) circle (2.5pt);
\draw[color=black] (88.76678131491627,8.733752621016578) node {$Y_1$};
\draw [fill=black] (90,20.080242834551548) circle (2.5pt);
\draw[color=black] (89.6248198108884,21.423690377235822) node {$Z_1$};
\draw [fill=black] (80,0) circle (2.5pt);
\draw[color=black] (81.24443485116717,1.1468859197894115) node {$X_0$};
\draw [fill=black] (100,0) circle (2.5pt);
\draw[color=black] (101.05027978356983,1.1468859197894115) node {$X_2$};
\draw [fill=black] (96,10) circle (2.5pt);
\draw[color=black] (96.76008730370921,8.733752621016578) node {$Y_3$};
\draw [fill=black] (100,20.080242834551548) circle (2.5pt);
\draw[color=black] (100.50836073348218,21.423690377235822) node {$Z_2$};
\draw [fill=black] (80,20) circle (2.5pt);
\draw[color=black] (80.59283564276079,21.423690377235822) node {$Z_0$};
\draw [fill=black] (84,10) circle (2.5pt);
\draw[color=black] (84.70238843925885,8.733752621016578) node {$Y_0$};
\draw [fill=black] (90,0) circle (2.5pt);
\draw[color=black] (90.48285830686052,1.1468859197894115) node {$X_1$};
\end{scriptsize}
\end{tikzpicture}
	\caption{The pseudo-tour $T^\leftarrow$.}
	\label{TleftM}
	\end{minipage}
\end{figure}
We also define the sets of pseudo-tours $T^\circ:=\{T_0^\circ,\dots, T_j^\circ\}$ and $T^\downarrow:=\{T_0^\downarrow, \dots, T_i^\downarrow\}$. Each of the tours $T_l^\circ$ (Figure \ref{TmidM}) and $T_l^\downarrow$ are defined similarly: Instead of $\{Z_s,Z_{s+1}\}$ we take two copies of the edges $\{Y_s,Y_{s+1}\}$ and $\{X_s,X_{s+1}\}$ except $\{Y_l,Y_{l+1}\}$ and $\{X_l,X_{l+1}\}$, respectively.

The tour $T^\nwarrow$ consists of (Figure \ref{TnwM}):
\begin{itemize}
\item two copies of the edges $\{Z_s,Z_{s+1}\}$ for all $0\leq s \leq k$
\item a copy of the edges $\{Y_s,Y_{s+1}\}$ for all $0\leq s \leq j$
\item a copy of the edges $\{X_s,X_{s+1}\}$ for all $0\leq s \leq i$
\item a copy of the edges $\{Z_0,Y_0\}$, $\{Z_0,X_0\}$ and $\{Y_{j+1},X_{i+1}\}$
\end{itemize}

The pseudo-tour $T^\nearrow$ is defined similarly. Instead of the edges $\{Z_0,Y_0\}$, $\{Z_0,X_0\}$ and $\{Y_{j+1},X_{i+1}\}$ it consists of the edges $\{Z_{k+1},Y_{j+1}\}$, $\{Z_{k+1},X_{i+1}\}$ and $\{Y_{0},X_{0}\}$.
We also define the pseudo-tours $T^\leftarrow$ (Figure \ref{TleftM}), $T^\rightarrow, T^\swarrow, T^\searrow$ similarly where we double the edges $\{Y_s,Y_{s+1}\}$ or $\{X_s,X_{s+1}\}$ instead of $\{Z_s,Z_{s+1}\}$.

We observe that the optimal tours of the instances maximizing the integrality ratio for $6\leq n\leq 10$ computed in Section \ref{sec computation rect} are the non-intersecting shortcuts of the pseudo-tours in $\mathfrak{T}$. In the next subsection the instances $I_{i,j,k}^2$ are constructed such that the vertices lie on three lines, are symmetric and the optimal tours are the non-intersection shortcuts of pseudo-tours in $\mathfrak{T}$. 
Similar properties also hold for other variants of the TSP. For example, the shortcuts of the pseudo-tours in $\mathfrak{T}$ are the optimal tours of the instances maximizing the integrality ratio in the metric case for $6\leq n \leq 12$ given in \cite{benoit2008finding}.

\subsection{The Instance $I^2_{i,j,k}$} \label{sub sec instance Iijk}
We define an embedding of $x_{i,j,k}$ in $\R^2$ (Figure \ref{Gijk}) as follows: The vertices $\{X_0, \dots, X_{i+1}\}$, $\{Y_0, \dots, Y_{j+1}\}$ and $\{Z_0, \dots, Z_{k+1}\}$ lie on the three parallel lines $l_1,l_2$ and $l_3$, respectively. The line $l_2$ lies between $l_1$ and $l_3$ in the plane. Moreover, $l_1,l_2$ and $l_2,l_3$ have distances $b_1:=\frac{1}{2}+\frac{j+1}{j+3}\left(\frac{1}{k+1}-\frac{1}{2}\right)$ and $b_2:=\frac{1}{2}+\frac{j+1}{j+3}\left(\frac{1}{i+1}-\frac{1}{2}\right)$ to each other, respectively. The vertices $X_0,X_{i+1},Z_0$ and $Z_{k+1}$ form an axis-parallel rectangle with side lengths 1 and $1+\frac{j+1}{j+3}\left(\frac{1}{i+1}+\frac{1}{k+1}-1\right)$. Let $Y'_l$ and $Z'_l$ be the orthogonal projection of $Y_l$ and $Z_l$ to the line $X_0X_{i+1}$, respectively. We call a sequence of points $v_1, \dots, v_s$ an \emph{equidistant progression} if $\dist_1(v_l,v_{l+1})=\dist_1(v_1,v_2)$ for all $1 \leq l \leq s-1$. Then $X_0, Y'_0, Y'_1, \dots, Y'_{j+1}, X_{i+1}$ and $Y'_0,X_1,X_2, \dots, X_{i},Y_{j+1}'$ and $Y'_0,Z'_1,Z'_2, \dots, Z'_{k},Y'_{j+1}$ are three equidistant progressions.

Note that this embedding also defines a \textsc{Rectilinear TSP} instance we call $I_{i,j,k}^2$. The explicit coordinates of the vertices are given by:
\begin{align*}
X_0&=(0,0)\\
X_{i+1}&=(1,0)\\
Z_0&=\left(0, 1+\frac{j+1}{j+3}
\left(\frac{1}{i+1}+\frac{1}{k+1}-1 \right) \right)\\
Z_{k+1}&=\left(1, 1+\frac{j+1}{j+3}
\left(\frac{1}{i+1}+\frac{1}{k+1}-1 \right) \right)\\
X_s&=\left(s\cdot\frac{j+1}{j+3}\cdot\frac{1}{i+1}+\frac{1}{j+3},0\right) & \forall 1&\leq s \leq i, \\
Y_s&=\left(\frac{(s+1)}{j+3},\frac{1}{2}+\frac{j+1}{j+3}\left(\frac{1}{k+1}-\frac{1}{2}\right)\right) & \forall 0&\leq s \leq j+1,\\
Z_s&=\left(s\cdot\frac{j+1}{j+3}\cdot\frac{1}{k+1}+\frac{1}{j+3},1+\frac{j+1}{j+3}
\left(\frac{1}{i+1}+\frac{1}{k+1}-1\right)\right) & \forall 1&\leq s \leq k.
\end{align*}

We call the vertices $X_0, \dots, X_{i+1}, Z_0, \dots, Z_{k+1}$ the \emph{outer vertices} and $Y_0, \dots, Y_{j+1}$ the \emph{inner vertices}. 

Revising the instances that maximize the integrality ratio from Section \ref{sec computation rect} we see that they can be transformed to $I_{i,j,k}^2$ for some $i,j$ and $k$ by scaling, rotating and translating. Moreover, their optimal fractional tours are isomorphic to $x_{i,j,k}$ and their optimal tours are the non-intersecting shortcuts of the pseudo-tours in $\mathfrak{T}$ (Figure \ref{non intersect shortcuts}). The values of $i,j,k$ and the corresponding integrality ratios are listed in Table \ref{comp result}.

\begin{table}
\center
{\renewcommand*\arraystretch{1.2}
\begin{tabular}{l|l|l|l}
$n$ & Instance & Opt. frac. tour & Integrality ratio \\
\hline
6 & $I_{0,0,0}$ & $x_{0,0,0}$ & $\frac{18}{17}\approx 1.059$ \\
7 & $I_{0,1,0}$ & $x_{0,1,0}$ & $\frac{13}{12}\approx 1.083$ \\
8 & $I_{0,1,1}$ & $x_{0,1,1}$ & $\frac{34}{31}\approx 1.097$ \\
9 & $I_{0,2,1}$ & $x_{0,2,1}$ & $\frac{31}{28}\approx 1.107$ \\
10 & $I_{1,2,1}$ & $x_{1,2,1}$ & $\frac{28}{25}=1.120$
\end{tabular}
}
\caption{The instances maximizing the integrality ratio computed in Section \ref{sec computation rect} with their corresponding fractional optimal tours and integrality ratios.}
\label{comp result}
\end{table}

\begin{figure}[!htb]
\begin{minipage}[t]{0.32\textwidth}
	\centering
	\begin{tikzpicture}[line cap=round,line join=round,>=triangle 45,x=0.2cm,y=0.2cm]
\draw [line width=2pt] (84,10)-- (80,20);
\draw [line width=2pt] (80,20)-- (80,0);
\draw [line width=2pt] (100,20.080242834551548)-- (100,0);
\draw [line width=2pt] (100,0)-- (90,0);
\draw [line width=2pt] (90,0)-- (80,0);
\draw [line width=2pt] (84,10)-- (88,10);
\draw [line width=2pt] (88,10)-- (92,10);
\draw [line width=2pt] (92,10)-- (96,10);
\draw [line width=2pt] (90,20.080242834551548)-- (100,20.080242834551548);
\draw [line width=2pt] (96,10)-- (90,20.080242834551548);
\begin{scriptsize}
\draw [fill=black] (92,10) circle (2.5pt);
\node[label=below:$Y_2$] at (92,10) {};
\draw [fill=black] (88,10) circle (2.5pt);
\node[label=below:$Y_1$] at (88,10) {};
\draw [fill=black] (90,20.080242834551548) circle (2.5pt);
\node[label=above:$Z_1$] at (90,20.080242834551548) {};
\draw [fill=black] (80,0) circle (2.5pt);
\node[label=80:$X_0$] at (80,0) {};
\draw [fill=black] (100,0) circle (2.5pt);
\node[label=80:$X_2$] at (100,0) {};
\draw [fill=black] (96,10) circle (2.5pt);
\node[label=below:$Y_3$] at (96,10) {};
\draw [fill=black] (100,20.080242834551548) circle (2.5pt);
\node[label=above:$Z_2$] at (100,20.080242834551548) {};
\draw [fill=black] (80,20) circle (2.5pt);
\node[label=above:$Z_0$] at (80,20) {};
\draw [fill=black] (84,10) circle (2.5pt);
\node[label=below:$Y_0$] at (84,10) {};
\draw [fill=black] (90,0) circle (2.5pt);
\node[label=above:$X_1$] at (90,0) {};
\end{scriptsize}
\end{tikzpicture}
	\caption*{The non-intersecting shortcut of $T^\uparrow_1$ for $I_{1,2,1}$.} 
	\label{Tup}
	\end{minipage}
	\hfill%
	\begin{minipage}[t]{0.32\textwidth}
	\centering	
	\begin{tikzpicture}[line cap=round,line join=round,>=triangle 45,x=0.2cm,y=0.2cm]
\draw [line width=2pt] (84,10)-- (80,20);
\draw [line width=2pt] (100,0)-- (90,0);
\draw [line width=2pt] (90,0)-- (80,0);
\draw [line width=2pt] (92,10)-- (96,10);
\draw [line width=2pt] (90,20.080242834551548)-- (100,20.080242834551548);
\draw [line width=2pt] (80,20)-- (90,20.080242834551548);
\draw [line width=2pt] (84,10)-- (88,10);
\draw [line width=2pt] (80,0)-- (88,10);
\draw [line width=2pt] (100,0)-- (96,10);
\draw [line width=2pt] (92,10)-- (100,20.080242834551548);
\begin{scriptsize}
\draw [fill=black] (92,10) circle (2.5pt);
\node[label=below:$Y_2$] at (92,10) {};
\draw [fill=black] (88,10) circle (2.5pt);
\node[label=below:$Y_1$] at (88,10) {};
\draw [fill=black] (90,20.080242834551548) circle (2.5pt);
\node[label=above:$Z_1$] at (90,20.080242834551548) {};
\draw [fill=black] (80,0) circle (2.5pt);
\node[label=above:$X_0$] at (80,0) {};
\draw [fill=black] (100,0) circle (2.5pt);
\node[label=80:$X_2$] at (100,0) {};
\draw [fill=black] (96,10) circle (2.5pt);
\node[label=below:$Y_3$] at (96,10) {};
\draw [fill=black] (100,20.080242834551548) circle (2.5pt);
\node[label=above:$Z_2$] at (100,20.080242834551548) {};
\draw [fill=black] (80,20) circle (2.5pt);
\node[label=above:$Z_0$] at (80,20) {};
\draw [fill=black] (84,10) circle (2.5pt);
\node[label=below:$Y_0$] at (84,10) {};
\draw [fill=black] (90,0) circle (2.5pt);
\node[label=above:$X_1$] at (90,0) {};
\end{scriptsize}
\end{tikzpicture}
	\caption*{A non-intersecting shortcut of $T^\circ_2$ for $I_{1,2,1}$. All shortcuts of $T^\circ_2$ are non-intersecting.}
	\label{Tmid}	
	\end{minipage}
	\hfill%
	\begin{minipage}[t]{0.32\textwidth}
	\centering	
	\begin{tikzpicture}[line cap=round,line join=round,>=triangle 45,x=0.2cm,y=0.2cm]
\draw [line width=2pt] (84,10)-- (80,20);
\draw [line width=2pt] (100,0)-- (90,0);
\draw [line width=2pt] (90,0)-- (80,0);
\draw [line width=2pt] (92,10)-- (96,10);
\draw [line width=2pt] (90,20.080242834551548)-- (100,20.080242834551548);
\draw [line width=2pt] (80,20)-- (90,20.080242834551548);
\draw [line width=2pt] (84,10)-- (88,10);
\draw [line width=2pt] (100,0)-- (100,20.080242834551548);
\draw [line width=2pt] (80,0)-- (96,10);
\draw [line width=2pt] (88,10)-- (92,10);
\begin{scriptsize}
\draw [fill=black] (92,10) circle (2.5pt);
\node[label=250:$Y_2$] at (92,10) {};
\draw [fill=black] (88,10) circle (2.5pt);
\node[label=below:$Y_1$] at (88,10) {};
\draw [fill=black] (90,20.080242834551548) circle (2.5pt);
\node[label=above:$Z_1$] at (90,20.080242834551548) {};
\draw [fill=black] (80,0) circle (2.5pt);
\node[label=above:$X_0$] at (80,0) {};
\draw [fill=black] (100,0) circle (2.5pt);
\node[label=80:$X_2$] at (100,0) {};
\draw [fill=black] (96,10) circle (2.5pt);
\node[label=below:$Y_3$] at (96,10) {};
\draw [fill=black] (100,20.080242834551548) circle (2.5pt);
\node[label=above:$Z_2$] at (100,20.080242834551548) {};
\draw [fill=black] (80,20) circle (2.5pt);
\node[label=above:$Z_0$] at (80,20) {};
\draw [fill=black] (84,10) circle (2.5pt);
\node[label=below:$Y_0$] at (84,10) {};
\draw [fill=black] (90,0) circle (2.5pt);
\node[label=above:$X_1$] at (90,0) {};
\end{scriptsize}
\end{tikzpicture}
	\caption*{A non-intersecting shortcut of $T^\leftarrow$ for $I_{1,2,1}$. All shortcuts of $T^\leftarrow$ are non-intersecting.}
	\label{Tleft}	
	\end{minipage}
	\caption{Some of the non-intersecting shortcuts of pseudo-tours in $\mathfrak{T}$.}
	\label{non intersect shortcuts}
\end{figure}

We observe that for $I_{i,j,k}^2$ there is a unique non-intersecting shortcut of pseudo-tours in $T^\uparrow$ and $T^\downarrow$ while all shortcuts of pseudo-tours in $T^\circ, T^\leftarrow, T^\rightarrow$ are non-intersecting. Moreover, the non-intersecting shortcuts of $T^\nwarrow, T^\nearrow, T^\swarrow$ and $T^\searrow$ are also non-intersecting shortcuts of $T^\leftarrow$ and $T^\rightarrow$. Hence, the set of optimal tours are the non-intersecting shortcuts of $T^\uparrow,T^\downarrow$ and all shortcuts of $T^\circ, T^\leftarrow, T^\rightarrow$.

For $n$ fixed let the instance $I^2_n:=I^2_{i^*,j^*,k^*}$ maximize the integrality ratio among all instances $I_{i,j,k}$ with $i+j+k+6=n$. 

\subsection{Length of the Optimal Tours for $I_{i,j,k}^2$}
In this subsection we determine the length of the optimal tours for $I_{i,j,k}^2$. This will be used to compute lower bounds on the integrality ratio of the instances in the next subsection.

A \emph{subpath} of an oriented tour consists of vertices $v_1,\dots,v_l$, such that $v_{i+1}$ is visited 
by the tour immediately after $v_i$ for all $i=1,\dots, l-1$. 
A subpath of a tour starting and ending at outer vertices and containing no other outer vertex is called a \emph{trip} 
if it contains at least one inner vertex. 

By Lemma \ref{lemma:convexhull}, we know that each optimal tour of $I_{i,j,k}^2$ can be decomposed into a set of trips and a set of edges connecting consecutive outer vertices such that all inner vertices are contained in some trip and two different trips intersect in at most one outer vertex.

\begin{lemma} \label{length optimal tour Iijk}
The length of the optimal tour for $I_{i,j,k}^2$ is $4+2b_1+2b_2-\frac{2}{j+3}$.
\end{lemma}

\begin{proof}
Assume that we have given an optimal tour. Since all inner vertices lie on a line, every trip visits a set of consecutive inner vertices. We start with a cycle visiting the outer vertices in cyclic order. This cycle has length
\begin{align*}
\dist_1(X_0,X_{i+1})+\dist_1(X_{i+1},Z_{k+1})+\dist_1(Z_{0},Z_{k+1})+\dist_1(X_0,Z_0)=2+2b_1+2b_2
\end{align*}
Now, we successively replace an edge between two consecutive outer vertices by a trip until every inner vertex is visited and we get the given optimal tour.

If we replace $\{Z_{r}, Z_{r+1}\}$ for some $1\leq r\leq k-1$ by a trip visiting the inner vertices $Y_s, Y_{s+1}, \dots, Y_t$, the cost of the cycle increases by at least
\begin{align*}
&\dist_1(Z_r,Y_s)+\dist_1(Y_s,Y_t)+\dist_1(Y_t,Z_{r+1})-\dist_1(Z_r,Z_{r+1})\\
\geq& 2\dist_1(Y_s,Y_t)-2\dist_1(Z_r,Z_{r+1})+2b_2\\
=& 2\dist_1(Y_s,Y_t)-2\cdot\frac{j+1}{j+3}\cdot\frac{1}{k+1}+1+\frac{j+1}{j+3}\left(\frac{2}{k+1}-1\right)=2\dist_1(Y_s,Y_t)+\frac{2}{j+3}
\end{align*}
Similarly, the cost increases by at least the same amount when we replace $\{X_r,X_{r+1}\}$ for some $1\leq r \leq k-1$ by a trip. 

In the case where we replace $\{Z_{0}, Z_{1}\}$ by a trip visiting the inner vertices $Y_s, Y_{s+1}, \dots, Y_t$ let $Y_0'$ be the orthogonal projection of $Y_0$ to the line $Z_0Z_{k+1}$. The cost of the cycle increases by at least
\begin{align*}
&\dist_1(Z_0,Y_s)+\dist_1(Y_s,Y_t)+\dist_1(Y_t,Z_{1})-\dist_1(Z_0,Z_{1})\\
\geq& 2\dist_1(Y_s,Y_t)-2\dist_1(Y'_0,Z_{1})+2b_2\\
=& 2\dist_1(Y_s,Y_t)-2\cdot\frac{j+1}{j+3}\cdot\frac{1}{k+1}+1+\frac{j+1}{j+3}\left(\frac{2}{k+1}-1\right)=2\dist_1(Y_s,Y_t)+\frac{2}{j+3}
\end{align*}
Similarly, the cost increases by at least the same amount when we replace one of the edges $\{Z_k,Z_{k+1}\}$, $\{X_0,X_{1}\}$ or $\{X_i,X_{i+1}\}$  by a trip. 

If we replace $\{X_0,Z_0\}$ by a trip visiting the vertices $Y_s, Y_{s+1}, \dots, Y_t$, the cost of the cycle increases by at least:
\begin{align*}
&\dist_1(X_0,Y_s)+\dist_1(Y_s,Y_t)+\dist_1(Y_t,Z_{0})-\dist_1(X_0,Z_{0})\\
=&\dist_1(X_0,Y_t)+\dist_1(Z_0,Y_t)-\dist_1(X_0,Z_0)\\
=&(t+1)\frac{1}{j+3}+b_1+(t+1)\frac{1}{j+3}+b_2 -(b_1+b_2)= 2(t+1)\frac{1}{j+3}\\
\geq& 2\dist(Y_s,Y_t)+\frac{2}{j+3}
\end{align*}
Similarly, we get the same value when we replace $\{X_{i+1},Z_{k+1}\}$ by a trip.
Assume that the optimal tour has exactly $w$ trips $t_1,\dots, t_w$. Since all the inner vertices are visited by the $w$ trips, all except $w-1$ edges of the form $\{Y_s,Y_{s+1}\}$ for $0\leq s\leq j$ are contained in the trips. Hence the total length of the optimal tour is at least:
\begin{align*}
&2+2b_1+2b_2+2\dist(Y_0,Y_{j+1})-(w-1)\cdot 2 \cdot \frac{1}{j+3}+w \frac{2}{j+3}\\
=&2+2b_1+2b_2+2-\frac{4}{j+3}+\frac{2}{j+3}=4+2b_1+2b_2-\frac{2}{j+3}
\end{align*}
In fact a straightforward calculation shows that all non-intersecting shortcuts of $\mathfrak{T}$ are optimal tours. For example the non-intersecting shortcut of the tour $T^\leftarrow$ has length:
\begin{align*}
&\dist_1(X_0,Y_{j+1})+\dist_1(Y_0,Y_{j+1})+\dist_1(Y_{0},Z_{0})+\dist_1(Z_0,Z_{k+1})+\dist_1(Z_{k+1},X_{i+1})\\
+&\dist_1(X_{i+1},X_0)=(b_1+1-\frac{1}{j+3})+(1-2\cdot \frac{1}{j+3})+(b_2+\frac{1}{j+3})+1+(b_1+b_2)+1\\
=&4+2b_1+2b_2-\frac{2}{j+3}
\end{align*}
Hence, the lower bound on the length of the optimal tour is tight.
\end{proof}

\subsection{The Integrality Ratio of $I_{i,j,k}^2$}
In this section we investigate the integrality ratio of $I_{i,j,k}^2$. Recall that $I^2_n$ is defined as the instance of the form $I_{i,j,k}$ with $n$ vertices and maximal integrality ratio.

\begin{theorem}
The integrality ratio of $I^2_{i,j,k}$ is at least $1+\frac{1}{3+2\left(\frac{5}{j+1}+\frac{1}{k+1}+\frac{1}{i+1}\right)}$. In particular, the integrality ratios of instances $I^2_{n}$ converge to $\frac{4}{3}$ as $n\to \infty$.
\end{theorem}

\begin{proof}
The cost of the optimal fractional tour of $I_{i,j,k}$ is at most the cost of $x_{i,j,k}$ which is
\begin{align*}
&\dist_1(X_0,X_{i+1})+\dist_1(Y_0,Y_{j+1})+\dist_1(Z_0,Z_{k+1})+\frac{1}{2}\dist_1(X_0,Y_0)+\frac{1}{2}\dist_1(X_0,Z_0)\\
+&\frac{1}{2}\dist_1(Y_0,Z_0)+\frac{1}{2}\dist_1(X_{i+1},Y_{j+1})+\frac{1}{2}\dist_1(X_{i+1},Z_{k+1})+\frac{1}{2}\dist_1(Y_{j+1},Z_{k+1})\\
=&3+2b_1+2b_2.
\end{align*}
By Lemma \ref{length optimal tour Iijk}, the cost of the optimal tour of $I_{i,j,k}$ is $4+2b_1+2b_2-\frac{2}{j+3}$. Hence, the integrality ratio is at least
\begin{align*}
&\frac{4+2b_1+2b_2-\frac{2}{j+3}}{3+2b_1+2b_2}=1+\frac{1-\frac{2}{j+3}}{3+2b_1+2b_2}\\
=&1+\frac{1-\frac{2}{j+3}}{3+1+\frac{j+1}{j+3}\left(\frac{2}{k+1}-1\right)+1+\frac{j+1}{j+3}\left(\frac{2}{i+1}-1\right)}\\
=&1+\frac{j+3-2}{5(j+3)+(j+1)\left(\frac{2}{k+1}+\frac{2}{i+1}-2\right)}=1+\frac{j+1}{3j+13+(j+1)\left(\frac{2}{k+1}+\frac{2}{i+1}\right)}\\
=&1+\frac{1}{3+2\left(\frac{5}{j+1}+\frac{1}{k+1}+\frac{1}{i+1}\right)}
\end{align*}

To get the highest integrality ratio of the instances $I_{i,j,k}^2$ we have to find $\min_{i,j,k} \frac{5}{j+1}+\frac{1}{i+1}+\frac{1}{k+1}$ where $i+j+k=n-6$ is fixed.

We get the following estimate by the Cauchy-Schwarz inequality:
\begin{align*}
\left(\frac{5}{j+1}+\frac{1}{i+1}+\frac{1}{k+1}\right)(n-3)&=\left(\frac{5}{j+1}+\frac{1}{i+1}+\frac{1}{k+1}\right)(j+1+i+1+k+1)\\
&\geq(\sqrt{5}+\sqrt{1}+\sqrt{1})^2 =(\sqrt{5}+2)^2
\end{align*}
With equality if and only if $\frac{5}{(j+1)^2}=\frac{1}{(i+1)^2}=\frac{1}{(k+1)^2}$. In this case the actual integrality ratio is
\begin{align*}
1+\frac{1}{3+2\frac{(\sqrt{5}+2)^2}{n-3}}\in \frac{4}{3}-\Theta(n^{-1})
\end{align*}
We see that this expression converges to $\frac{4}{3}$ as $n\to \infty$. However, the optimal values of $i, j$ and $k$ we chose above do not have to be integral. So in order to conclude that the actual integrality ratios of the instance $I^2_n=I_{i,j,k}^2$ for the best choice of $i,j,k$ also converge to $\frac{4}{3}$ we need to show that the optimal integral values of $i,j,k$ achieve a similar integrality ratio. Let $i', j'$ and $k'$ be real numbers satisfy $\frac{5}{(j')^2}=\frac{1}{(i')^2}=\frac{1}{(k'+1)^2}$ and $i'+j'+k'=n-6$. With the restriction that $i,j$ and $k$ are integers we can choose $i=\lfloor i' \rfloor, j=\lfloor j' \rfloor$ and $k=\lceil k' \rceil$ or $k=\lceil k' \rceil+1$ such that $i+j+k=i'+j'+k'=n-6$. In this case we have by a similar application of the Cauchy-Schwarz inequality as above:
\begin{align*}
1+\frac{1}{3+2\left(\frac{5}{j+1}+\frac{1}{k+1}+\frac{1}{i+1}\right)}&\geq 1+\frac{1}{3+2\left(\frac{5}{j'}+\frac{1}{k'+1}+\frac{1}{i'}\right)}= 1+\frac{1}{3+2\frac{(\sqrt{5}+2)^2}{n-5}} \\
&\in \frac{4}{3}-\Theta(n^{-1}).
\end{align*}
\end{proof}

\section{Integrality Ratio for Metric TSP} \label{sec int ratio metric}
In this section we give an upper bound on the integrality ratio for \textsc{Metric TSP} instances whose optimal fractional tour is isomorphic to $x_{i,j,k}$ for some $i,j,k$ with $i+j+k+6=n$. This implies that, assuming Conjecture \ref{structure conjecture}, the \textsc{Metric TSP} instances described in \cite{benoit2008finding} maximize the integrality ratio. Note that in \cite{BOYD2011525} it was already shown that these instances have integrality ratios that are upper bounded by $\frac{4}{3}$. We start by defining coefficients for the pseudo-tours in $\mathfrak{T}$.

\begin{definition}
We define the real coefficients
\begin{align*}
\lambda^\uparrow&:=\frac{1}{\lvert T^\uparrow \rvert} \cdot\frac{1}{3+2(\frac{1}{i+1}+\frac{1}{j+1}+\frac{1}{k+1})}\\
\lambda^\circ&:=\frac{1}{\lvert T^\circ \rvert} \cdot \frac{1}{3+2(\frac{1}{i+1}+\frac{1}{j+1}+\frac{1}{k+1})}\\
\lambda^\downarrow&:=\frac{1}{\lvert T^\downarrow \rvert} \cdot \frac{1}{3+2(\frac{1}{i+1}+\frac{1}{j+1}+\frac{1}{k+1})}\\
\lambda^\nwarrow&:=\lambda^\nearrow:=\frac{\frac{1}{k+1}}{3+2(\frac{1}{i+1}+\frac{1}{j+1}+\frac{1}{k+1})}\\
\lambda^\leftarrow&:=\lambda^\rightarrow:=\frac{\frac{1}{j+1}}{3+2(\frac{1}{i+1}+\frac{1}{j+1}+\frac{1}{k+1})}\\
\lambda^\swarrow&:=\lambda^\searrow:=\frac{\frac{1}{i+1}}{3+2(\frac{1}{i+1}+\frac{1}{j+1}+\frac{1}{k+1})}.
\end{align*}
\end{definition}

\begin{lemma} \label{optimal gap metric}
The integrality ratio of \textsc{Metric TSP} instances whose optimal fractional tour isomorphic to $x_{i,j,k}$ is at most $1+\frac{1}{3+2(\frac{1}{i+1}+\frac{1}{j+1}+\frac{1}{k+1})}$.
\end{lemma}

\begin{proof}
For a pseudo-tour $T$ let $\chi^T$ be the vector such that $\chi^T(e)$ is the number of occurrence of $e$ in $T$ for all $e\in E(K_n)$. In order to show the statement, we show that
\begin{align*}
&\sum_{T\in T^\uparrow}\lambda^\uparrow\chi^T+\sum_{T\in T^\circ}\lambda^\circ\chi^T+\sum_{T\in T^\downarrow}\lambda^\downarrow\chi^T+\lambda^\nwarrow\chi^{T^\nwarrow}+\lambda^\nearrow\chi^{T^\nearrow}+\lambda^\leftarrow\chi^{T^\leftarrow}+\lambda^\rightarrow\chi^{T\rightarrow}+\lambda^\swarrow\chi^{T^\swarrow}\\
+&\lambda^\searrow\chi^{T^\searrow}=\left(1+\frac{1}{3+2(\frac{1}{i+1}+\frac{1}{j+1}+\frac{1}{k+1})} \right) x_{i,j,k}.
\end{align*}
This implies the Lemma since
\begin{align*}
&\sum_{T\in T^\uparrow}\lambda^\uparrow+\sum_{T\in T^\circ}\lambda^\circ+\sum_{T\in T^\downarrow}\lambda^\downarrow+\lambda^\nwarrow+\lambda^\nearrow+\lambda^\leftarrow+\lambda^\rightarrow+\lambda^\swarrow+\lambda^\searrow\\
=&3\cdot \frac{1}{3+2(\frac{1}{i+1}+\frac{1}{j+1}+\frac{1}{k+1})}+ 2 \cdot \frac{\frac{1}{k+1}}{3+2(\frac{1}{i+1}+\frac{1}{j+1}+\frac{1}{k+1})}+2\cdot \frac{\frac{1}{j+1}}{3+2(\frac{1}{i+1}+\frac{1}{j+1}+\frac{1}{k+1})}\\
&+2\cdot \frac{\frac{1}{i+1}}{3+2(\frac{1}{i+1}+\frac{1}{j+1}+\frac{1}{k+1})}=1.
\end{align*}
Consider the edge $\{Z_l,Z_{l+1}\}$ for some $0\leq l \leq k$. It is contained in every pseudo-tour of $T^\circ, T^\downarrow, T^\leftarrow, T^\rightarrow, T^\swarrow, T^\searrow$ once and in each of $T^\nwarrow, T^\nearrow$ twice. In the pseudo-tours of $T^\uparrow$ it is contained twice except of the tour $T^\uparrow_l$ where it is not contained. Hence,
\begin{align*}
&\sum_{T\in T^\uparrow}\lambda^\uparrow\chi^T(\{Z_l,Z_{l+1}\})+\sum_{T\in T^\circ}\lambda^\circ\chi^T(\{Z_l,Z_{l+1}\})+\sum_{T\in T^\downarrow}\lambda^\downarrow\chi^T(\{Z_l,Z_{l+1}\})\\
+&\lambda^\nwarrow\chi^{T^\nwarrow}(\{Z_l,Z_{l+1}\})+\lambda^\nearrow\chi^{T^\nearrow}(\{Z_l,Z_{l+1}\})+\lambda^\leftarrow\chi^{T^\leftarrow}(\{Z_l,Z_{l+1}\})\\
+&\lambda^\rightarrow\chi^{T\rightarrow}(\{Z_l,Z_{l+1}\})+\lambda^\swarrow\chi^{T^\swarrow}(\{Z_l,Z_{l+1}\})+\lambda^\searrow\chi^{T^\searrow}(\{Z_l,Z_{l+1}\})\\
=&2\lambda^\uparrow-\frac{2}{\lvert T^\uparrow \rvert}\lambda^\uparrow+\lambda^\circ+\lambda^\downarrow+2\lambda^\nwarrow+2\lambda^\nearrow+\lambda^\leftarrow+\lambda^\rightarrow+\lambda^\swarrow+\lambda^\searrow\\
=&4\cdot \frac{1}{3+2(\frac{1}{i+1}+\frac{1}{j+1}+\frac{1}{k+1})}-\frac{2}{k+1}\cdot \frac{1}{3+2(\frac{1}{i+1}+\frac{1}{j+1}+\frac{1}{k+1})}+4 \cdot \frac{\frac{1}{k+1}}{3+2(\frac{1}{i+1}+\frac{1}{j+1}+\frac{1}{k+1})}\\
&+2\cdot \frac{\frac{1}{j+1}}{3+2(\frac{1}{i+1}+\frac{1}{j+1}+\frac{1}{k+1})}+2\cdot \frac{\frac{1}{i+1}}{3+2(\frac{1}{i+1}+\frac{1}{j+1}+\frac{1}{k+1})}\\
=&4\cdot \frac{1}{3+2(\frac{1}{i+1}+\frac{1}{j+1}+\frac{1}{k+1})}+2 \cdot \frac{\frac{1}{k+1}}{3+2(\frac{1}{i+1}+\frac{1}{j+1}+\frac{1}{k+1})}+2\cdot \frac{\frac{1}{j+1}}{3+2(\frac{1}{i+1}+\frac{1}{j+1}+\frac{1}{k+1})}\\
&+2\cdot \frac{\frac{1}{i+1}}{3+2(\frac{1}{i+1}+\frac{1}{j+1}+\frac{1}{k+1})}\\
=&1+\frac{1}{3+2(\frac{1}{i+1}+\frac{1}{j+1}+\frac{1}{k+1})}=\left(1+\frac{1}{3+2(\frac{1}{i+1}+\frac{1}{j+1}+\frac{1}{k+1})}\right)x_{i,j,k}(\{Z_l,Z_{l+1}\}).
\end{align*}
Next, consider the edge $\{Z_0,Y_0\}$. It is contained in the pseudo-tours of $T^\uparrow, T^\circ, T^\nwarrow, T^\leftarrow, T^\searrow$ once and not contained in the pseudo-tours of $T^\downarrow, T^\nearrow, T^\rightarrow, T^\swarrow$.  Hence,
\begin{align*}
&\sum_{T\in T^\uparrow}\lambda^\uparrow\chi^T(\{Z_0,Y_0\})+\sum_{T\in T^\circ}\lambda^\circ\chi^T(\{Z_0,Y_0\})+\sum_{T\in T^\downarrow}\lambda^\downarrow\chi^T(\{Z_0,Y_0\})\\
+&\lambda^\nwarrow\chi^{T^\nwarrow}(\{Z_0,Y_0\})+\lambda^\nearrow\chi^{T^\nearrow}(\{Z_0,Y_0\})+\lambda^\leftarrow\chi^{T^\leftarrow}(\{Z_0,Y_0\})\\
+&\lambda^\rightarrow\chi^{T\rightarrow}(\{Z_0,Y_0\})+\lambda^\swarrow\chi^{T^\swarrow}(\{Z_0,Y_0\})+\lambda^\searrow\chi^{T^\searrow}(\{Z_0,Y_0\})\\
=&\lambda^\uparrow+\lambda^\circ+\lambda^\nwarrow+\lambda^\leftarrow+\lambda^\searrow\\
=&2\cdot \frac{1}{3+2(\frac{1}{i+1}+\frac{1}{j+1}+\frac{1}{k+1})}+\frac{\frac{1}{k+1}}{3+2(\frac{1}{i+1}+\frac{1}{j+1}+\frac{1}{k+1})}+\frac{\frac{1}{j+1}}{3+2(\frac{1}{i+1}+\frac{1}{j+1}+\frac{1}{k+1})}\\
&+\frac{\frac{1}{i+1}}{3+2(\frac{1}{i+1}+\frac{1}{j+1}+\frac{1}{k+1})}\\
=&\frac{1}{2}\left(1+\frac{1}{3+2(\frac{1}{i+1}+\frac{1}{j+1}+\frac{1}{k+1})}\right)=\left(1+\frac{1}{3+2(\frac{1}{i+1}+\frac{1}{j+1}+\frac{1}{k+1})}\right)x_{i,j,k}(\{Z_0,Y_0\}).
\end{align*}
The statement can be shown for all other edges of $x_{i,j,k}$ analogously to one of the two cases above.
\end{proof}
\begin{remark}
Theorem 4.1 in \cite{benoit2008finding} shows that the upper bound in Lemma \ref{optimal gap metric} is tight, i.e.\ there is actually an instance where the integrality ratio is equal to $1+\frac{1}{3+2(\frac{1}{i+1}+\frac{1}{j+1}+\frac{1}{k+1})}$.
\end{remark}

\begin{theorem} \label{max int ratio metric}
The integrality ratio of \textsc{Metric TSP} instances whose optimal fractional tour is isomorphic to $x_{i,j,k}$ with $i+j+k+6=n$ is at most
\begin{align*}
\begin{cases}
  1+\frac{1}{3+\frac{18}{n-3}} & \text{if } n\equiv 0 \mod 3\\
  1+\frac{1}{3+2(\frac{6}{n-4}+\frac{3}{n-1})} & \text{if } n\equiv 1 \mod 3\\
  1+\frac{1}{3+2(\frac{3}{n-5}+\frac{6}{n-2})} & \text{if } n\equiv 2 \mod 3.
 \end{cases}
\end{align*}
\end{theorem}

\begin{proof}
In order to maximize the integrality ratio, we need to find $i,j,k$ with $i+j+k+6=n$ that minimize $\frac{1}{i+1}+\frac{1}{j+1}+\frac{1}{k+1}$. Since the function $f(x)=\frac{1}{x}$ is convex, we can use Jensen's inequality to see that the integrality ratio is maximized for $i+1=j+1=k+1$ if $n$ is divisible by 3. For $n\equiv 1 \mod 3$ and $n\equiv 2 \mod 3$ we use the Karamata's inequality (Theorem \ref{Karamata}) to find the best values of $i,j$ and $k$. For $n\equiv 1 \mod 3$ it is maximized for $i+1=j+1=k$ since this triple is been majorized by all other integer triples. Similarly, for $n\equiv 2 \mod 3$ the integrality ratio is maximized for $i+1=j=k$.
\end{proof}

\begin{remark}
Conjecture 4.1 in \cite{benoit2008finding} states the given bounds in Theorem \ref{max int ratio metric} hold for arbitrary \textsc{Metric TSP} instances with $n$ vertices and is tight. Thus, Conjecture \ref{structure conjecture} would imply Conjecture 4.1 in \cite{benoit2008finding}.
\end{remark}

\section{Integrality Ratio for Multidimensional Rectilinear TSP} \label{sec int ratio rectilinear multi}
In this section we show that there are \textsc{Rectilinear TSP} instances in $\R^3$ that have the same integrality ratio as the upper bounds given in Theorem \ref{max int ratio metric} for the \textsc{Metric TSP}. Hence, assuming Conjecture \ref{structure conjecture} the exact integrality ratio for \textsc{Multidimensional Rectilinear TSP} is the same as in the metric case. Since the instances in $\R^3$ can be embedded into $\R^d$ for $d\geq 3$, the statement also holds for these spaces.

We start by constructing an instance $I^3_{i,j,k}$ with the vertex set of $x_{i,j,k}$. The coordinates of the vertices are given by $X_s=(0,0,\frac{s}{i+1}), Y_s=(\frac{1}{i+1}+\frac{1}{j+1},0,\frac{s}{j+1})$ and $Z_s=(\frac{1}{i+1},\frac{1}{k+1},\frac{s}{k+1})$. The vertices form a prism in the three dimensional space where the triangle $X_0,Y_0,Z_0$ lies in the plane $z=0$ and the triangle $X_{i+1},Y_{j+1},Z_{k+1}$ lies in the plane $z=1$. The sequences of vertices $X_0,\dots, X_{i+1}$ and $Y_0,\dots,Y_{j+1}$ and $Z_0,\dots, Z_{k+1}$ are equidistant progressions such that each sequence lies on one of three parallel lines.  

We can check that $\dist_1(X_l,X_{l+1})=\frac{1}{i+1}, \dist_1(Y_l,Y_{l+1})=\frac{1}{j+1}$ and $\dist_1(Z_l,Z_{l+1})=\frac{1}{k+1}$. Moreover, $\dist_1(X_0,Y_0)=\frac{1}{i+1}+\frac{1}{j+1}, \dist_1(Y_0,Z_0)=\frac{1}{j+1}+\frac{1}{k+1}$ and $\dist_1(X_0,Z_0)=\frac{1}{i+1}+\frac{1}{k+1}$. The same distances also hold for the triangle $X_{i+1},Y_{j+1},Z_{k+1}$. Note these distances are the same as the corresponding distances of the \textsc{Metric TSP} instances Benoit and Boyd described in \cite{benoit2008finding}. Nevertheless, it is not clear that they have the same integraltiy ratio, since the remaining distances are given by the 1-norm instead of the metric closure of the weighted graph.

For all fixed $n$ let the instance $I^3_n:=I_{i^*,j^*,k^*}^3$ be the instance that maximizes the integrality ratio among the instances $I^3_{i,j,k}$ with $i+j+k+6=n$. Next, we determine the length of the optimal tours of $I^3_{i,j,k}$.

\begin{lemma} \label{optimal tour 3D bound}
Every optimal tour of $I_{i,j,k}^3$ has at least length $4+\frac{2}{i+1}+\frac{2}{j+1}+\frac{2}{k+1}$.
\end{lemma}

\begin{proof}
Assume that we have given an optimal tour $T$. We call an edge \emph{vertical} if it is parallel to the $z$-axis, otherwise it is called \emph{non-vertical}. A vertical edge is called \emph{base edge} if it connects two consecutive vertices $\{X_{s},X_{s+1}\}, \{Y_s,Y_{s+1}\}$ or $\{Z_s,Z_{s+1}\}$. We may assume that all vertical edges are base edges, otherwise we can replace them by a set of base edges to get a pseudo-tour with equal length. We call all base edges which are not in $T$ \emph{gaps}. Furthermore, we add auxiliary vertices $X_{-1}=(0,0,-\frac{1}{i+1})$, $X_{i+2}=(0,0,\frac{i+2}{i+1})$, $Y_{-1}=(\frac{1}{i+1}+\frac{1}{j+1},0,-\frac{1}{j+1})$, $Y_{j+2}=(\frac{1}{i+1}+\frac{1}{j+1},0,\frac{j+2}{j+1})$ and $Z_{-1}=(\frac{1}{i+1},\frac{1}{k+1},-\frac{1}{k+1})$, $Z_{k+2}=(\frac{1}{i+1},\frac{1}{k+1},\frac{k+2}{k+1})$ such that the sequences $X_{-1},X_0,\dots,X_{i+2}$ and $Y_{-1},Y_0, \dots,Y_{j+2}$ and $Z_{-1},Z_0,\dots, Z_{k+2}$ are equidistant progressions. We call the edges $\{X_{-1},X_0\}$, $\{X_{i+1},X_{i+2}\}$, $\{Y_{-1},Y_{0}\}$, $\{Y_{j+1},Y_{j+2}\}$, $\{Z_{-1},Z_0\}$ and $\{Z_{k+1},Z_{k+2}\}$ the \emph{auxiliary gaps}. 

Next, we assign the non-vertical edges in $T$ to the gaps such that every non-vertical edge is assigned to two gaps incident to the edge on different lines, every non-auxiliary gap is assigned to two non-vertical edges and every auxiliary gap is assigned to one non-vertical edge. We do this as follows: For any endpoint of a non-vertical edge that is only incident to one gap we assign this edge to that gap. For any endpoint of a non-vertical edge incident to two gaps there has to be another non-vertical edge incident to that endpoint. We arbitrarily assign one of the edges to a gap and the other to the other gap. Note that since every gap has two endpoints we assigned two non-vertical edges to it this way. Moreover, every auxiliary gap has only one endpoint which is vertex of $x_{i,j,k}$ and hence it is assigned to one non-vertical edge.

Since we use the Manhattan norm, we can replace every edge $e\in T$ by three edges we call the \emph{subedges} of $e$ without changing the length of the tour such that they are parallel to the $x$-, $y$- and $z$-axis, respectively. After the replacement, we also call the subedges parallel to the $z$-axis \emph{vertical}. Next, we replace all non-vertical subedges originated from a non-vertical edges of $T$ by the two gaps it is assigned to and get a multiset of edges $T'$. We claim that $T'$ has the same length as $T$. To see this assume that the non-vertical edge $\{X_l,Y_s\}$ is in $T$ and note that the non-vertical subedges of it have total length $\dist_1(X_0,Y_0)=\frac{1}{i+1}+\frac{1}{j+1}$. Assume that we assigned this edge to the gaps $\{X_l,X_{l+1}\}$ and $\{Y_s, Y_{s+1}\}$. Then, these two edges we add have also total length $\frac{1}{i+1}+\frac{1}{j+1}$. Similar statements hold for the non-vertical edges of the form $\{X_l,Z_s\}$ and $\{Y_l,Z_s\}$. Hence, $T'$ has the same length as $T$. 

Now, for $a\in \R$ consider the intersection of the plane $z=a$ with $T'$. We claim that it intersects $T'$ at least four times for all $0 \leq a\leq 1$. The plane intersects each of the segments $X_0X_{i+1}$, $Y_0Y_{j+1}$ and $Z_0Z_{k+1}$ at a base edge or a gap. If it intersects at least one gap, the statement is true since every non-auxiliary gap was assigned to two non-vertical edges and we filled the gap by two edges. Otherwise, the plane intersects three base edges of $T$. Since the pseudo-tour $T$ intersects a plane an even number of times, it has to intersect at least 4 times. Thus, it also intersects $T'$ at least 4 times since we only replaced the non-vertical subedges. This shows the claim and implies that the part of the edges in $T'$ with $z$-coordinate between 0 and 1 has length at least 4.

Moreover, we added an edge in $T'$ to every auxiliary gap. Since the interior of the auxiliary gaps does not have $z$-coordinates between 0 and 1, this increases the lower bound of the length of $T'$ by $\frac{2}{i+1}+\frac{2}{j+1}+\frac{2}{k+1}$. Therefore, the total length of $T'$ and thus also that of $T$ is at least $4+\frac{2}{i+1}+\frac{2}{j+1}+\frac{2}{k+1}$.
\end{proof}

\begin{corollary} \label{3d embedding same ratio}
The integrality ratio of $I^3_{i,j,k}$ is at least $1+\frac{1}{3+2(\frac{1}{i+1}+\frac{1}{j+1}+\frac{1}{k+1})}$.
\end{corollary}

\begin{proof}
By Lemma \ref{optimal tour 3D bound}, the length of every tour is at least $4+\frac{2}{i+1}+\frac{2}{j+1}+\frac{2}{k+1}$. The length of the fractional tour $x_{i,j,k}$ is
\begin{align*}
&\dist_1(X_0,X_{i+1})+\dist_1(Y_0,Y_{j+1})+\dist_1(Z_0,Z_{k+1})+\frac{1}{2}\dist_1(X_0,Y_0)+\frac{1}{2}\dist_1(X_0,Z_0)\\
+&\frac{1}{2}\dist_1(Y_0,Z_0)+\frac{1}{2}\dist_1(X_{i+1},Y_{j+1})+\frac{1}{2}\dist_1(X_{i+1},Z_{k+1})+\frac{1}{2}\dist_1(Y_{j+1},Z_{k+1})\\
=&3+\frac{2}{i+1}+\frac{2}{j+1}+\frac{2}{k+1}.
\end{align*}
Hence, the integrality ratio is at least
\begin{align*}
\frac{4+\frac{2}{i+1}+\frac{2}{j+1}+\frac{2}{k+1}}{3+\frac{2}{i+1}+\frac{2}{j+1}+\frac{2}{k+1}}=1+\frac{1}{3+2(\frac{1}{i+1}+\frac{1}{j+1}+\frac{1}{k+1})}.
\end{align*}
\end{proof}

\begin{corollary}
The instance $I^3_{i,j,k}$ has the same or higher integrality ratio as any \textsc{Metric TSP} instance whose optimal fractional solution is isomorphic to $x_{i,j,k}$. 
\end{corollary}

\begin{proof}
By Lemma \ref{optimal tour 3D bound}, any \textsc{Metric TSP} instance with $x_{i,j,k}$ as the optimal fractional tour has at most integrality ratio $1+\frac{1}{3+2(\frac{1}{i+1}+\frac{1}{j+1}+\frac{1}{k+1})}$. The instance $I^3_{i,j,k}$ has at least the same integrality ratio by Corollary \ref{3d embedding same ratio}.
\end{proof}

\begin{remark}
Corollary \ref{3d embedding same ratio} implies an alternative proof of the following statement: There exists a \textsc{Metric TSP} instance with $i+j+k+6$ vertices having an integrality ratio of $1+\frac{1}{3+2(\frac{1}{i+1}+\frac{1}{j+1}+\frac{1}{k+1})}$.
This is a direct consequence of Lemma 4.1 in \cite{benoit2008finding} and Lemma 4.2 in \cite{benoit2008finding} and the key ingredient of Theorem 4.1 in \cite{benoit2008finding}.
\end{remark}

\section{Local Optimality} \label{sec local opt} \label{sec local optimal criterion}
In this section we consider TSP instances with $n$ vertices that can be embedded into $\R^d$ such that the distances arise from a norm which is totally differentiable in every non-zero point. An instance is called locally optimal if its integrality ratio cannot be increased by making small changes to its embedded vertices. We describe a criterion to check if an instance is locally optimal and develop an algorithm that finds locally optimal instances.

\subsection{A Criterion for Local Optimality}
Assume that we have given a norm $\Vert \cdot \Vert$ in $\R^d$ which is totally differentiable in every non-zero point. Let a TSP instance $(K_n,c)$ be given where the vertices can be embedded into $\R^d$ as the vertex set $v=\{v_1,\dots,v_n\}$ with $v_1, \dots, v_n \in \R^d$ such that the cost function $c$ arises from $v$ and the given norm in $\R^d$. Since the instance is completely characterized by the embedded vertex set $v$, we will simply call this instance $v$. W.l.o.g.\ assume that no two vertices of the instance coincide, i.e. $v_i\neq v_j$ for $i\neq j$. We can interpret the vertex set $v$ as a point in $\R^{nd}$ since each of the $n$ vertices is a point in $\R^d$. For another point  $w\in \R^{nd}$ and $\lambda\in \R$ we define a new instance $v+\lambda w$ where we add and multiply the coordinates of the vertices coordinate-wise. Moreover, let $\mathbb{T}$ be the set of optimal tours and $\mathbb{X}$ be the set of optimal fractional tours for $v$. 

\begin{definition}
For a given tour $T$ and fractional tour $x$ and instance $y\in \R^{nd}$ let $l_T,l_x:\R^{nd}\to\R$ be defined as follows: $l_T(y)$ and $l_x(y)$ denote the length of $T$ and cost of $x$ for the instance $y$, respectively. Moreover, let $r_{T,x}:\R^{nd}\to \R$ be defined as $r_{T,x}(y):=\frac{l_T(y)}{l_x(y)}$, the ratio of the length of $T$ and the cost of $x$ for the instance $v$. Furthermore, let $r_y$ be the integrality ratio of the instance $y$.
\end{definition} 

Now, we can define local optimality for a TSP instance.

\begin{definition}
The instance $v$ is called \emph{locally optimal} if there does not exist $w\in \R^{nd}$ such that $\liminf_{\epsilon\to 0}\frac{r_{v+\epsilon w}-r_v}{\epsilon}>0$. 
\end{definition}

\begin{lemma} \label{cond local optimal}
The instance $v$ is locally optimal if and only if for some $x\in \mathbb{X}$ there does not exist $w\in \R^{nd}$ such that $\partial_wr_{T,x}(v)> 0$ for all $T\in \mathbb{T}$.
\end{lemma}

\begin{proof}
Since the length of the tour and fractional tour is continuous in $v$ and the number of tours and fractional tours is finite, given $w$ for small enough $\epsilon$ the optimal tour of the instance $v+\epsilon w$ is still in $\mathbb{T}$ and the optimal fractional tour is still in $\mathbb{X}$. Hence, the new integrality ratio is 
\begin{align*}
r_{v+\epsilon w}=\frac{\min_{T\in \mathbb{T}}l_T(v+\epsilon w)}{\min_{x \in \mathbb{X}}l_x(v+\epsilon w)}=\max_{x \in \mathbb{X}}\min_{T\in \mathbb{T}}\frac{l_T(v+\epsilon w)}{l_x(v+\epsilon w)}=\max_{x \in \mathbb{X}}\min_{T\in \mathbb{T}}r_{T,x}(v+\epsilon w).
\end{align*}
Note that since the given norm is differentiable the functions $l_T,l_x$ and $r_{T,x}$ are also differentiable. Moreover, we have by definition $r_v=r_{T,x}(v)$ for all $T\in \mathbb{T}$ and $x\in \mathbb{X}$. Therefore, $\liminf_{\epsilon\to 0}\frac{r_{v+\epsilon w}-r_v}{\epsilon}>0$ if and only if we have $\lim_{\epsilon\to 0}\frac{r_{T,x}(v+\epsilon w)- r_{T,x}(v)}{\epsilon}=\partial_w r_{T,x}(v)>0$ for some $x\in \mathbb{X}$ and all $T\in \mathbb{T}$.
\end{proof}

For a tour $T$ and a fractional tour $x$ let $\delta_{T}(v_i)$ and $\delta_{x}(v_i)$ denote the set of vertices incident to $v_i$ in $T$ and the support graph of $x$, respectively. 

\begin{lemma} We have
\begin{align*}
\partial_wl_T(v)&=\frac{1}{2}\sum_{i\in \{1,\dots, n\}} \sum_{v_j\in \delta_{T}(v_i)} \partial_w \Vert v_j-v_i \Vert\\
\partial_wl_x(v)&=\frac{1}{2}\sum_{i\in \{1,\dots, n\}} \sum_{v_j\in \delta_{x}(v_i)} x(\{v_i,v_j\}) \partial_w \Vert v_j-v_i \Vert.
\end{align*}
\end{lemma}

\begin{proof}
Note that the length of the tour and fractional tour can be expressed as 
\begin{align*}
l_T(v)&=\frac{1}{2}\sum_{i\in \{1,\dots,n\}}\sum_{v_j\in \delta_{T}(v_i)}\Vert v_j-v_i \Vert\\
l_x(v)&=\frac{1}{2} \sum_{i\in \{1,\dots,n\}}\sum_{v_j\in \delta_{x}(v_i)}x(\{v_i,v_j\}) \Vert v_j- v_i \Vert.
\end{align*}
The statement follows from the fact that the derivative is linear.
\end{proof}

\begin{definition}
Define the function $g_{T,x,v}:\R^{nd} \to \R$ as $g_{T,x,v}(y):=l_T(y)-r_vl_x(y)$. 
\end{definition}

\begin{lemma} \label{optcriterion}
The instance $v$ is locally optimal if and only if for some $x\in \mathbb{X}$ there does not exist $w\in \R^{nd}$ such that $\partial_{w}g_{T,x,v}(v)> 0$ for all $T\in \mathbb{T}$.
\end{lemma}

\begin{proof}
By Lemma \ref{cond local optimal} $v$ is locally optimal if and only if for some $x\in \mathbb{X}$ there does not exist $w\in \R^{nd}$ such that $\partial_wr_{T,x}(v)>0$ for all $T\in \mathbb{T}$. We have:
\begin{align*}
\partial_wr_{T,x}(v)=\partial_w\frac{l_T(v)}{l_x(v)}=\frac{(\partial_wl_T(v))l_x(v)-l_T(v)(\partial_wl_x(v))}{l_x(v)^2}.
\end{align*}
Hence
\begin{align*}
\partial_wr_{T,x}(v)> 0 &\Leftrightarrow \frac{(\partial_wl_T(v))l_x(v)-l_T(v)(\partial_wl_x(v))}{l_x(v)^2}> 0\\
&\Leftrightarrow (\partial_wl_T(v))l_x(v)-l_T(v)(\partial_wl_x(v))> 0\\
&\Leftrightarrow \partial_wl_T(v)- \frac{l_T(v)}{l_x(v)}\partial_wl_x(v)> 0\\
&\Leftrightarrow \partial_w(l_T(v)- r_vl_x(v))> 0\\
&\Leftrightarrow \partial_wg_{T,x,v}(v)> 0.
\end{align*}
\end{proof}

\begin{lemma} \label{optcriterion2}
The instance $v$ is locally optimal if and only if for some $x\in \mathbb{X}$ there does not exist $w\in \R^{nd}$ such that $\langle w, \nabla g_{T,x,v}(v) \rangle> 0$ for all $T\in \mathbb{T}$.
\end{lemma}

\begin{proof}
Recall that we assumed that the norm is totally differentiable for $p>1$ in any non-zero point. Therefore, for $w\in \R^{nd}$ we have $\partial_wg_{T,x,v}(v)=\langle w, \nabla g_{T,x,v}(v) \rangle$. The statement follows from Lemma \ref{optcriterion}.
\end{proof}

\begin{theorem}
The instance $v$ is locally optimal if and only if for some $x\in \mathbb{X}$ there exist $\{\lambda_T \geq 0\}_{T\in \mathbb{T}}$ not all zero such that $\sum_{T \in \mathbb{T}}\lambda_T \nabla g_{T,x,v}(v)=\vec{0}$ where $\vec{0}$ is the vector consisting of zeros. 
\end{theorem}

\begin{proof}
Consider the following LP:
\begin{align}
\min \ &0 \nonumber\\
s.t. \ \langle w,\nabla g_{T,x,v}(v) \rangle \geq &0 \qquad  \forall T\in \mathbb{T} \label{improvement LP}
\end{align}
and its dual LP
\begin{align*}
\max \ &0\\
s.t. \ \sum_{T\in \mathbb{T}}\lambda_T\nabla g_{T,x,v}(v)= &\vec{0}\\
\lambda_T \geq &0 \qquad \forall T\in \mathbb{T}
\end{align*}
Note that all feasible solutions are optimal and both systems are feasible since we can set all variables equal to zero. By Lemma \ref{optcriterion2}, the instance $v$ is not locally optimal if and only if the primal has a solution where all inequalities are not tight. By complementary slackness, the dual has in this case only solutions where all $\lambda_T$ are zero. Moreover, if the dual has a solution where $\lambda_T\neq 0$ for some $T\in \mathbb{T}$, by complementary slackness the corresponding inequality of the primal is tight for any primal solution.
\end{proof}

\subsection{Local Optimality for the $p$-Norm}
In this subsection we apply the criterion from the last subsection to the $p$-norm for $p>1$ explicitly, i.e.\ we choose $\Vert \cdot \Vert = \Vert \cdot \Vert_p$. Note that the $p$-norm is differentiable in every non-zero point for $p>1$ and hence satisfies the condition for the criterion.

Using a straightforward calculation with the chain rule we get for every unit vector $e\in \R^{nd}$:

\begin{align*}
\partial_e \Vert v_i - v_j \Vert_p=\sgn(\langle v_i-v_j,e \rangle) \frac{\lvert\langle v_i-v_j,e \rangle\rvert^{p-1}}{\Vert v_i-v_j \Vert^{p-1}_p}
\end{align*}
where 
\begin{align*}
\sgn(x):=\begin{cases*} 1 & if $x\geq 0$\\0 & if $x=0$\\ -1& else \end{cases*}
\end{align*}
is the sign function of $x$.

Thus, we have:
\begin{align*}
\partial_el_T(v)&=\frac{1}{2}\sum_{i\in \{1,\dots,n\}} \sum_{q\in \delta_{T}(v_i)} \sgn(\langle v_i-q,e \rangle) \frac{\lvert\langle v_i-q,e \rangle\rvert^{p-1}}{\Vert v_i-q \Vert^{p-1}}\\
\partial_el_x(v)&=\frac{1}{2}\sum_{i\in \{1,\dots,n\}} \sum_{q\in \delta_{x}(v_i)} x(\{v_i,q\})\sgn(\langle v_i-q,e \rangle) \frac{\lvert\langle v_i-q,e \rangle\rvert^{p-1}}{\Vert v_i-q \Vert^{p-1}}
\end{align*}

With this we can compute $\nabla g_{T,x,v}$ as $\nabla g_{T,x,v}=\nabla l_T +g_v\nabla l_x$.

\begin{remark}
The above criterion can also be applied in the case of the 1-norm in restricted form. The 1-norm is totally differentiable for every instance where no two vertices have equal coordinates at the same position, i.e. in the two-dimensional case no two vertices have the same $x$- or $y$-coordinate. Hence, the criterion can be applied in these cases. For instances where there exist multiple vertices with equal coordinates at the same position we can treat these coordinates as one variable. In this case it is differentiable again and the criterion can be applied. Note that if the criterion is applied that way it does not necessarily detect all instances that are not locally optimal.
\end{remark}

\subsection{A Local Search Algorithm} \label{sec local search int ratio}
In this subsection we develop a local search algorithm that finds a local optimal solution with respect to the integrality ratio.

Note that for instances that are not locally optimal the LP (\ref{improvement LP}) has a solution where all inequalities are not tight. Therefore, we can solve a slightly modified LP that gives a direction vector that can be added to the instance to improve the integrality ratio.

We start by generating random instances until we get an instance $v$ with integrality ratio by a given constant greater than 1. In every iteration we solve the following LP and try to improve the current integrality ratio:

\begin{align}
\max \ &\delta \nonumber\\
s.t. \ \langle w,\nabla g_{T,x,v}(v) \rangle \geq &\delta \qquad  \forall T\in \mathbb{T} \label{local search LP}\\
-1 \leq  w_i \leq &1 \qquad \forall i\in\{1,\dots, nd\} \nonumber
\end{align}

If the objective value $\delta$ is greater than zero, $w$ corresponds to a solution of LP (\ref{improvement LP}) where all inequalities are not tight. Note that we added bounds for $w_i$ to ensure that the LP is bounded. Given an optimal solution $w$ of the LP we use binary search to determine the maximal $\eta$ such that $v+\eta w$ has higher integrality ratio than $v$. We maintain a list $\mathbb{T}$ of optimal or near-optimal tours. In each iteration we include the current optimal tour $T^*$ to $\mathbb{T}$ and delete the tours that are by more than a given constant longer than $T^*$. In contrast to the optimal tours we only store one current optimal fractional tour since in practice the local optima has usually many optimal tours but a unique optimal fractional tour (Algorithm \ref{local search for integrality ratio}).

\begin{algorithm}
\caption{Local Search Algorithm for Integrality Ratio}
\label{local search for integrality ratio}
 \hspace*{\algorithmicindent} \textbf{Input:} Number of vertices $n$, accuracy parameters $\epsilon_0,\epsilon_1,\epsilon_2,\epsilon_3>0$ \\
 \hspace*{\algorithmicindent} \textbf{Output:} Locally optimal instance $v$
\begin{algorithmic}[1]
\Do
\State Generate a random instance $v$ with $n$ vertices
\doWhile{integrality ratio of $v$ is smaller than $1+\epsilon_0$}
\State Compute an optimal tour $T^*$ and an optimal fractional tour $x^*$ of $v$
\State Let $\mathbb{T}:=\{T^*\}$
\While{LP (\ref{local search LP}) has a solution $w$ with objective value $>\epsilon_1$}
\State Find by binary search $\eta$ maximal such that $g(v+\eta w)>g(v)$
\If {$\eta<\epsilon_2$}
\State \textbf{break}
\EndIf
\State Let $v:=v+\eta w$
\State Compute an optimal tour $T^*$ and an optimal fractional tour $x^*$ of $v$
\State Set $\mathbb{T}:=\mathbb{T}\cup \{T^*\}$
\State Delete all tours in $\mathbb{T}$ that is at least $\epsilon_3$ longer than $T^*$
\EndWhile
\State \Return $v$
\end{algorithmic}
\end{algorithm}

\section{Integrality Ratio for Euclidean TSP} \label{sec int ratio euclidean}
In this section, we investigate the integrality ratio of \textsc{Euclidean TSP}. Using the local search algorithm described in Section \ref{sec local search int ratio} we can find local optima with respect to the integrality ratio for \textsc{Euclidean TSP} instances. Unfortunately, there are many such local optima. This means that we had to restart the algorithm several times with a small random modification of the last local optimum to get good results. The instances we found in the end with the highest integrality ratio seem to have the following structural properties that share similarities with the instances maximizing the integrality ratio in the rectilinear and metric case:
\begin{observation} \label{property euclidean}
We observe the following properties for the instances with the highest integrality ratio for \textsc{Euclidean TSP} found by the local search algorithm:
\begin{enumerate}
\item The optimal fractional solution is isomorphic to $x_{i,j,k}$ for some $i,j,k$. \label{prop support}
\item The non-intersecting shortcuts of the pseudo-tours in $\mathfrak{T}$ are optimal tours (Figure \ref{non intersect shortcuts}). \label{prop shortest tours}
\end{enumerate}
If the optimal fractional solution is isomorphic to $x_{i,j,i}$ for some $i,j$, we further obtain the following properties:
\begin{enumerate}
\setcounter{enumi}{2}
\item The instance can be rotated and shifted such that it is symmetric to the $x$- and $y$-axis and the inner vertices lie on the $x$-axis. \label{prop symmetry}
\item The outer vertices lie on an ellipse with foci on the $x$-axis. \label{prop ellipse}
\end{enumerate}
\end{observation}

In the following we will refer to these properties by property \ref{prop support}, \ref{prop shortest tours}, \ref{prop symmetry} and \ref{prop ellipse}. Using them we develop an efficient algorithm that constructs instances with high integrality ratio.

\subsection{The Ellipse Construction Algorithm}
In this subsection we describe an algorithm we call the \emph{ellipse generation algorithm} that efficiently generates instances satisfying Observation \ref{property euclidean}.

In the following we assume that $i=k$ to use the additional properties \ref{prop symmetry} and \ref{prop ellipse} for an efficient algorithm to construct instances that match these patterns. This lets us construct instances with a larger number of vertices since due to the high number of local optima it is quite time consuming to use the local search method to generate good instances. The instances generated by the new efficient algorithm have for all tested $n$ high integrality ratio. Like in the rectilinear case there are values of $n$ where the optimal fractional solution is isomorphic to $x_{i,j,i+1}$ for some $i,j$, for example $n=8,11$ (Figure~\ref{optEuc8}). Nevertheless, we ignore these cases since these instances are less symmetric and we do not understand their structure well. 

\begin{figure}
\centering
\begin{minipage}[t]{0.46\textwidth}
	\centering
 \includegraphics[trim=70 0 0 0,scale=0.6]{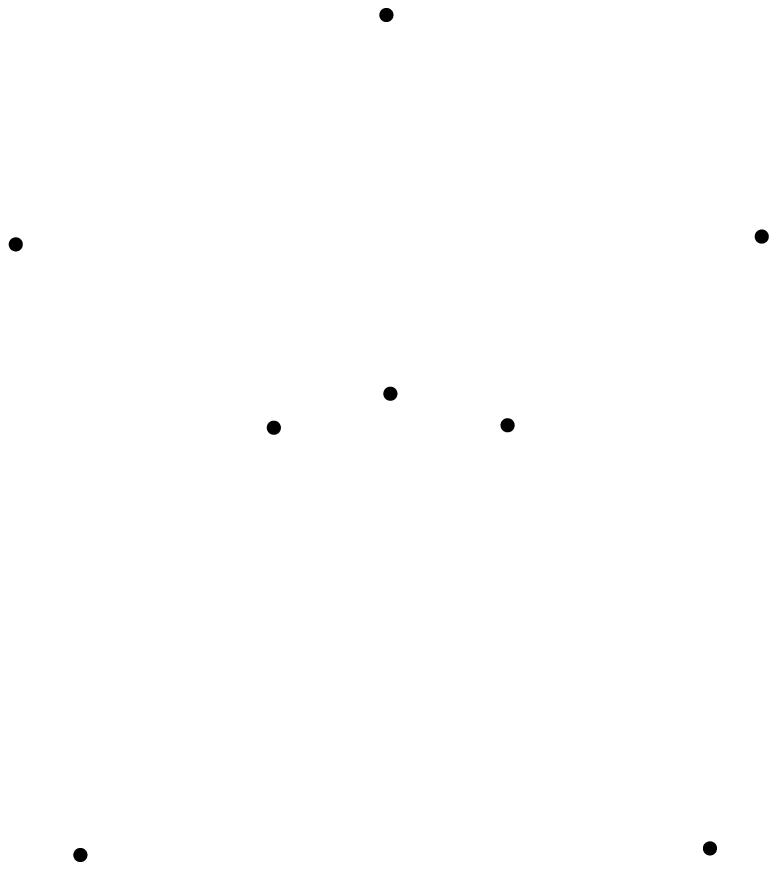}
	\captionsetup{labelformat=empty}
	\caption{$n=8$}
	\end{minipage}
	\hfill%
	\begin{minipage}[t]{0.46\textwidth}
	\centering	
 \includegraphics[trim=70 0 0 0,scale=0.6]{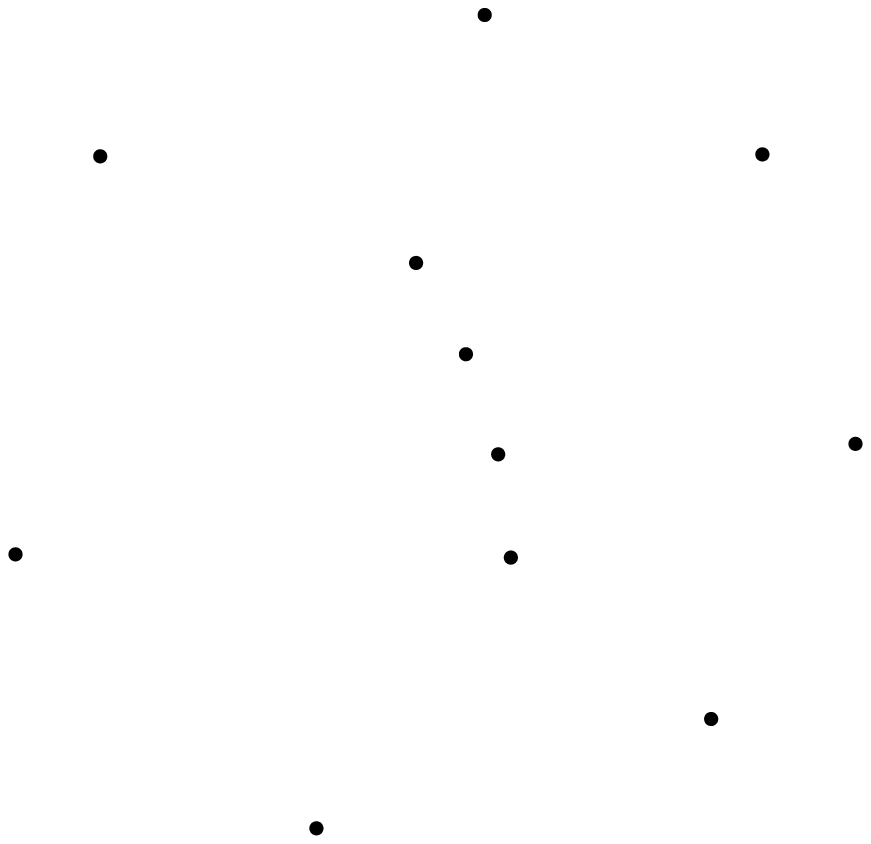}
	\captionsetup{labelformat=empty}
	\caption{$n=11$}	
	\end{minipage} 
 
  \caption{The instances with the highest integrality ratio for $n=8,11$ found by the local search algorithm. Their optimal fractional tours are isomorphic to $x_{1,1,0}$ and $x_{2,2,1}$, respectively. The integrality ratios of these instances are approximately 1.0435 and 1.0789, respectively. The best integrality ratios found by the ellipse generation algorithm are approximately 1.0413 and 1.0784 achieved by instances which optimal fractional tour are isomorphic to $x_{1,0,1}$ and $x_{1,3,1}$, respectively.}
  \label{optEuc8}
\end{figure}

By property \ref{prop support}, we may assume that the optimal fractional tour is isomorphic to $x_{i,j,i}$ with some fixed $i,j$. By symmetry, we can w.l.o.g.\ assume that the vertices $X_0,X_{i+1},Z_0,Z_{i+1}$ have the coordinates $(-b,-1),(b,-1),(-b,1),(b,1)$ for some $b>0$, respectively. The explicit value of $b$ will be chosen later in the procedure. In the following we assume that an explicit value of $b$ is given and describe how to determine the coordinates of the inner and outer vertices.

\subsubsection*{Inner Vertices}
We first compute the coordinates of the inner vertices. By property \ref{prop symmetry}, we know that the inner vertices lie on the $x$-axis and are symmetric to $(0,0)$. We first set the coordinates of $Y_0$ to $(-f,0)$ where $f$ is a parameter to be determined later. With a given value for $f$ the coordinates of $Y_{h+1}$ for increasing $h$ can be iteratively determined as follows: Assume we already know the coordinates of $Y_h$. Consider the difference of arbitrary shortcuts of the pseudo-tours $T^\leftarrow$ and $T^\circ_h$ (by symmetry all shortcuts have the same length) which is 
\begin{align*}
\diffInner_h:=&\dist_2(X_0,Y_{j+1})+\dist_2(X_{i+1},Z_{i+1})+\dist_2(Y_h,Y_{h+1})-\dist_2(X_0,Y_h)\\
-&\dist_2(Y_{h+1},Z_{i+1})-\dist_2(X_{i+1},Y_{j+1})
\end{align*}
(Figure \ref{inner diff}). Since the coordinates of $X_0,X_{i+1},Z_{i+1},Y_h$ and $Y_{j+1}$ are already known, this is equivalent to $\diffInner_h:=\dist_2(Y_h,Y_{h+1})-\dist_2(Y_{h+1},Z_{i+1})+c$ for some constant $c$. Since $Y_{h+1}$ lies on the $x$-axis right of $Y_{h}$ and left of $Z_{i+1}$ the condition $\diffInner_h=0$ from property \ref{prop shortest tours} determines the position of $Y_{h+1}$ uniquely as $\diffInner_h$ is monotonic increasing in the $x$-coordinate of $Y_{h+1}$. By symmetry, this also determines the position of $Y_{j+1-(h+1)}$. 

\begin{figure}
\centering
\definecolor{ffqqqq}{rgb}{1,0,0}
\begin{tikzpicture}[line cap=round,line join=round,>=triangle 45,x=2cm,y=2cm]
\draw [line width=2pt] (-0.5761520080509908,0)-- (-0.18897861556520823,0);
\draw [line width=2pt] (1.9894031203497522,1)-- (1.9894031203497522,-1);
\draw [line width=2pt,color=ffqqqq] (1.9894031203497522,-1)-- (0.9973071205038739,0);
\draw [line width=2pt,color=ffqqqq] (-0.18897861556520823,0)-- (1.9894031203497522,1);
\draw [line width=2pt] (-1.9894031203497522,-1)-- (0.9973071205038739,0);
\draw [line width=2pt,color=ffqqqq] (-1.9894031203497522,-1)-- (-0.5761520080509908,0);
\begin{scriptsize}
\draw [fill=black] (-1.9894031203497522,1) circle (2pt);
\draw[color=black] (-1.9446839190869623,1.097071007979368) node {$Z_0$};
\draw [fill=black] (1.9894031203497522,1) circle (2pt);
\draw[color=black] (2.056288670280544,1.137071007979368) node {$Z_{i+1}$};
\draw [fill=black] (-1.9894031203497522,-1) circle (2pt);
\draw[color=black] (-2.0085398407127775,-0.8784715673192991) node {$X_0$};
\draw [fill=black] (1.9894031203497522,-1) circle (2pt);
\draw[color=black] (2.236288670280544,-0.9024175379289799) node {$X_{i+1}$};
\draw [fill=black] (-0.18897861556520823,0) circle (2pt);
\draw[color=black] (-0.22855602539316885,0.11927720808406819) node {$Y_2$};
\draw [fill=black] (-0.5761520080509908,0) circle (2pt);
\draw[color=black] (-0.5797635943351545,0.11927720808406819) node {$Y_1$};
\draw [fill=black] (-0.9973071205038739,0) circle (2pt);
\draw[color=black] (-0.9788631044965016,0.11927720808406819) node {$Y_0$};
\draw [fill=black] (0.9973071205038739,0) circle (2pt);
\draw[color=black] (1.0625308899787895,0.13932223257600084) node {$Y_{j+1}$};
\end{scriptsize}
\end{tikzpicture}
  \caption{Construction of $Y_{h+1}$ for $h+1=2$. The edges shown are that of $T^\leftarrow \triangle T^\circ_h$. The condition $\diffInner_h=0$ gives that the length of the red edges is equal to that of the black edges.}
  \label{inner diff}
\end{figure}
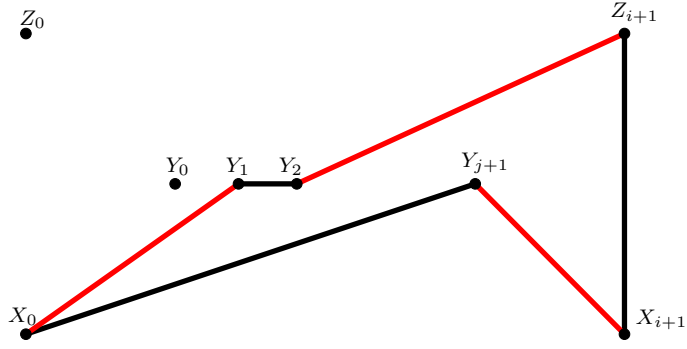

For odd $j$ we further know that the coordinate of $Y_{\frac{j+1}{2}}$ is by symmetry $(0,0)$. Therefore, we need to find $f$ such that $\diffInner_\frac{j-1}{2}=0$. Similarly, for even $j$ the coordinate of $Y_{\frac{j}{2}+1}$ is determined by symmetry and we need to find $f$ such that $\diffInner_{\frac{j}{2}}=0$. Now, we could try every value of $f$ up to a certain accuracy. To speed up the calculation, we make the following observation:

\begin{observation}
The $x$-coordinate of $Y_l$ for a fixed $l$ with $1\leq l \leq j$ determined by the procedure above is monotonically decreasing in $f$.
\end{observation}

This observation can be shown using induction and monotonicity arguments. Hence, we can use binary search to find the correct $f$ such that the inner vertices are symmetric to the $y$-axis (Algorithm \ref{Inner Vertices Algorithm}).

\begin{algorithm}[H]
\caption{Inner Vertices Algorithm}
\label{Inner Vertices Algorithm}
 \hspace*{\algorithmicindent} \textbf{Input:} $i,j,b$, accuracy $\epsilon$ \\
 \hspace*{\algorithmicindent} \textbf{Output:} Coordinates of the inner vertices $Y_0,\dots, Y_{j+1}$
\begin{algorithmic}[1]
\If{$j$ is odd}:
\State Set the coordinates of the vertex $Y_{\frac{j+1}{2}}:=(0,0)$
\EndIf
\State \textbf{Binary search} for $f$ such that the inner vertices satisfy Observation \ref{property euclidean}:
\For{\textbf{each} value $f$ we evaluate}
\State Set the coordinates of $Y_0:=(-f,0)$ and $Y_{j+1}:=(f,0)$
\For {$h$ from 0 to $\lfloor\frac{j}{2} \rfloor-1$}
\State Determine the coordinates of $Y_{h+1}$
\State Set by symmetry the coordinates of $Y_{j+1-(h+1)}$
\EndFor
\If{$\lvert \diffInner_{\lfloor \frac{j}{2} \rfloor} \rvert<\epsilon$}
\State \Return coordinates of $Y_0,\dots, Y_{j+1}$
\Else
\State Increase or decrease $f$ depending on the sign of $\diffInner_{\lfloor \frac{j}{2} \rfloor}$ and repeat
\EndIf
\EndFor
\end{algorithmic}
\end{algorithm}

\subsubsection{Outer Vertices}
The outer vertices are harder to compute since they are not uniquely determined by $b$ in contrast to the inner vertices. By property \ref{prop ellipse}, we know that $X_0,X_{i+1},Z_0$ and $Z_{i+1}$ and the other outer vertices lie on an ellipse. Unfortunately, five vertices are needed to determine an ellipse, therefore we need one more parameter. By property \ref{prop symmetry}, the coordinate of the foci of the ellipse have the form $(-e,0)$ and $(e,0)$. We will take the value of $e$ as an additional parameter.

Using the ellipse we make the observation that the coordinate of the other outer vertices are uniquely determined. We iteratively construct the coordinates of $Z_{h+1}$ for increasing $h$. Assume that the coordinates of $Z_h$ are already known. Consider the difference of the length of any non-intersecting shortcuts of $T^\leftarrow$ and $T^\uparrow_h$ (by symmetry all of the non-intersecting shortcuts have the same length). It is equal to 
\begin{align*}
\diffOuter_h:=&\dist_2(X_0,Y_0)+\dist_2(Z_0,Y_{j+1})-\dist_2(X_0,Z_0)-\dist_2(Z_h,Y_0)\\
-&\dist_2(Z_{h+1},Y_{j+1})+\dist_2(Z_h,Z_{h+1}).
\end{align*}
Moreover, we know the coordinates of $X_0,Z_0, Z_h, Y_0$ and $Y_{j+1}$. By property \ref{prop shortest tours} we have $\diffOuter_h=0$ and therefore $\dist_2(Z_h,Z_{h+1})-\dist_2(Z_{h+1},Y_{j+1})+c=0$ for some constant $c$. This equation describes a hyperbola with foci $Z_h$ and $Y_{j+1}$. It intersects the ellipse twice, once on each side of the line $Z_hY_{j+1}$ (Figure \ref{intersection hyperbola ellipse}). Thus, we can determine $Z_{h+1}$ as the intersection of the hyperbola with the ellipse that lies on the same side as $Z_{i+1}$.

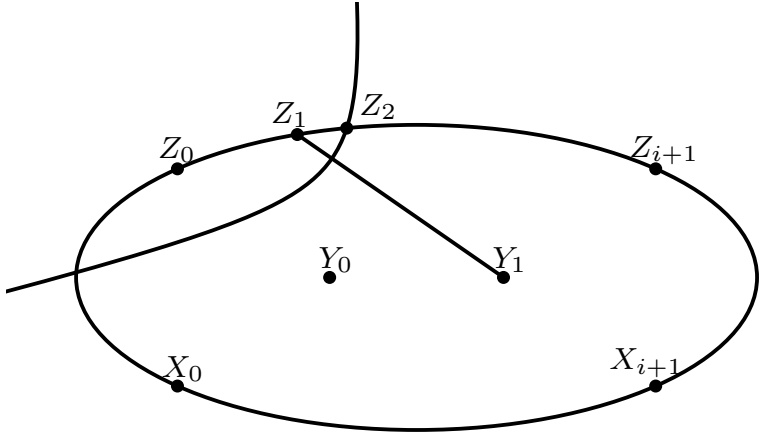
\begin{figure}
\centering
\resizebox{.7\textwidth}{!}{%
\begin{tikzpicture}[line cap=round,line join=round,>=triangle 45,x=1.cm,y=1.cm]
\clip(-3.77266940744422,-2.0423205570473986) rectangle (3.881616063018635,2.5329202675613285);
\draw [rotate around={0:(0,0)},line width=1pt] (0,0) ellipse (3.1326049462809213cm and 1.4047112690740755cm);
\draw [line width=1pt] (-1.0975232623670848,1.3156765228537806)-- (0.7997396721745,0);
\draw [samples=50,domain=-0.99:0.99,rotate around={145.26022782461806:(-0.1488917950962924,0.6578382614268903)},xshift=-0.1488917950962924cm,yshift=0.6578382614268903cm,line width=1pt] plot ({0.7668762603316844*(1+(\x)^2)/(1-(\x)^2)},{0.8628752170697573*2*(\x)/(1-(\x)^2)});
\begin{scriptsize}
\draw [fill=black] (2.2,1) circle (1.5pt);
\draw[color=black] (2.2785487286386785,1.1718227387055875) node {$Z_{i+1}$};
\draw [fill=black] (2.2,-1) circle (1.5pt);
\draw[color=black] (2.116797790394935,-0.7838494355572719) node {$X_{i+1}$};
\draw [fill=black] (-2.2,1) circle (1.5pt);
\draw[color=black] (-2.2042629884022085,1.17915319637456) node {$Z_0$};
\draw [fill=black] (-2.2,-1) circle (1.5pt);
\draw[color=black] (-2.1464947961723,-0.842733531672226) node {$X_0$};
\draw [fill=black] (0.7997396721745,0) circle (1.5pt);
\draw[color=black] (0.8574511997829334,0.16243301312817615) node {$Y_1$};
\draw [fill=black] (-0.7997396721745,0) circle (1.5pt);
\draw[color=black] (-0.7427277249855276,0.16243301312817615) node {$Y_0$};
\draw [fill=black] (-1.0975232623670848,1.3156765228537806) circle (1.5pt);
\draw[color=black] (-1.1644355282638588,1.4995477959700824) node {$Z_1$};
\draw [fill=black] (-0.6425331580457277,1.3748451724394826) circle (1.5pt);
\draw[color=black] (-0.3505390521808485,1.574646445868963) node {$Z_2$};
\end{scriptsize}
\end{tikzpicture}}
  \caption{Construction of $Z_{h+1}$ with $h+1=2$ and $j=0$ as the intersection of the ellipse with the hyperbola.}
  \label{intersection hyperbola ellipse}
\end{figure}

By symmetry, the coordinates of $Z_{i+1-(h+1)},X_{h+1},X_{i+1-(h+1)}$ are determined by that of $Z_{h+1}$. For odd $i$ by symmetry the coordinate of the vertex $Z_{\frac{i+1}{2}}$ has $x$-coordinate 0 and has positive $y$-coordinate. Since the vertex lies on the ellipse, its coordinates are uniquely determined. Therefore, we want to choose $e$ such that $\diffOuter_{\frac{i-1}{2}}=0$. For even $i$ the coordinates of $Z_{\frac{i}{2}+1}$ are determined by symmetry and we want to choose $e$ such that $\diffOuter_{\frac{i}{2}}=0$. It remains to determine the value(s) of $e$ such that the condition above is satisfied. We could simply try every value of $e$ up to a certain accuracy. To speed up the process, we make the following observations:

\begin{observation} \label{outer coordinate observation}
The $x$-coordinate of $Z_{\lfloor\frac{i}{2}\rfloor}$ determined by the procedure above is monotonically increasing in $e$.
\end{observation}

\begin{observation} \label{outer observation}
$\diffOuter_{\lfloor \frac{i}{2} \rfloor}$ is monotonically decreasing in $e$ if the vertex $Z_{\lfloor\frac{i}{2}\rfloor}$ constructed by the procedure above has negative $x$-coordinate. 
\end{observation}

Using Observation \ref{outer coordinate observation} and \ref{outer observation} we can speed up the process of finding the correct value of $e$ by using binary search on $e$ (Algorithm \ref{Outer Vertices Algorithm}).

\begin{algorithm}
\caption{Outer Vertices Algorithm}
\label{Outer Vertices Algorithm}
 \hspace*{\algorithmicindent} \textbf{Input:} $i,j,b$, coordinates for the vertices $Y_0,Y_{j+1}$, accuracy $\epsilon$ \\
 \hspace*{\algorithmicindent} \textbf{Output:} Coordinates for the vertices $X_1,\dots, X_{i}, Z_1,\dots, Z_{i}$
\begin{algorithmic}[1]
\State \textbf{Binary search} for $e$ such that the outer vertices are symmetric to the $y$-axis:
\For {\textbf{each} $e$ we evaluate}
\If{$i$ is odd}
\State Set $X_{\frac{i+1}{2}}$ and $Z_{\frac{i+1}{2}}$ to the intersections of $x=0$ with the ellipse through $X_0$ with foci $(-e,0)$ and $(e,0)$
\EndIf
\For{$h$ from 0 to $\lfloor\frac{i}{2}\rfloor-1$}
\State Compute the position of $Z_{h+1}$ as the intersection of the hyperbola $\diffOuter_h=0$ and the ellipse through the outer vertices with foci $(-e,0)$ and $(e,0)$
\State Set the position of $Z_{i+1-(h+1)}, X_{h+1}, X_{i+1-(h+1)}$ by symmetry
\EndFor
\If{$Z_{\lfloor \frac{i}{2}\rfloor}$ has non-negative $x$-coordinate}
\State Decrease $e$ and repeat
\EndIf
\If{$\lvert \diffOuter_{\lfloor\frac{i}{2} \rfloor} \rvert <\epsilon$}
\State \Return the coordinates of $X_1,\dots, X_i,Z_1,\dots, Z_i$
\Else
\State Increase or decrease $e$ depending on the sign of $\diffOuter_{\lfloor\frac{i}{2} \rfloor}$ and repeat
\EndIf
\EndFor 
\end{algorithmic}
\end{algorithm}

\subsubsection{Best Value for $b$}
Let us denote the integrality ratio of the instance constructed by the inner and outer vertices algorithm with $b$ as given parameter by $\ratio(b)$. Note that $\ratio(b)$ is not defined for every $b$: If $b$ is too small or too large, there is no $e$ such that the outer vertices can be constructed satisfying the properties \ref{prop shortest tours} and \ref{prop symmetry}. In this case the outer vertices algorithm fails to find suitable coordinates for the outer vertices. To find the instance with maximal integrality ratio by this construction we need to determine the value of $b$ that maximizes $\ratio(b)$. We make the following observation that helps us to do this efficiently:

\begin{observation}
The function $\ratio(b)$ is a concave function in $b$.
\end{observation}

Therefore, we can efficiently minimize a concave function to find the $b$ maximizing $\ratio(b)$ instead of using brute force. Given a constructed instance with a fixed $b$ the properties \ref{prop support} and \ref{prop shortest tours} allow us to speed up the computation of the integrality ratio: Instead of computing an optimal fractional tour and an optimal tour we can just compute the cost of $x_{i,j,i}$ and the length of any non-intersecting shortcut of a pseudo-tour in $\mathfrak{T}$. All in all, the above considerations result in the following algorithm we call the ellipse construction algorithm:

\begin{algorithm}[H]
\caption{\textsc{Euclidean TSP} Ellipse Construction Algorithm}
 \hspace*{\algorithmicindent} \textbf{Input:} $i,j$, accuracy $\epsilon$ \\
 \hspace*{\algorithmicindent} \textbf{Output:} \textsc{Euclidean TSP} instance
\begin{algorithmic}[1]
\State Optimize over a concave function to find $b$ that maximizes $\ratio(b)$:
\For{\textbf{each} $b$ we evaluate}
\State Use the inner vertices algorithm to compute the inner vertices
\State Use the outer vertices algorithm to compute the outer vertices
\State Compute the length of any non-intersecting shortcut of a pseudo-tour in $\mathfrak{T}$
\State Compute the cost of the fractional tour $x_{i,j,i}$
\State Divide the two values to get the integrality ratio assuming properties \ref{prop support} and \ref{prop shortest tours}.
\EndFor
\end{algorithmic}
\end{algorithm}

\subsection{Results of the Ellipse Construction Algorithm}
In this subsection we describe the instances found by the ellipse construction algorithm. For $\epsilon=10^{-9}$ and every $6\leq n\leq 199$, $x_{i,j,i}$ with $i+j+i+6=n$ we executed the algorithm and took the best result for every $n$. The actual integrality ratio of the resulting instances have been computed using the \texttt{Concorde} TSP solver for $6\leq n \leq 109$ vertices. The computation time was too high for the remaining instances, see Section \ref{sec runtime concorde} for more details on this phenomenon. Up to this point we could verify that property \ref{prop shortest tours} of Observation \ref{property euclidean} holds, i.e.\ the non-intersecting shortcuts of the pseudo-tours in $\mathfrak{T}$ are optimal tours. For higher number of vertices the integrality ratio was computed assuming property \ref{prop shortest tours}. Some of the instances the ellipse construction algorithm generated are shown in Figure \ref{results euclidean}.
\begin{figure}[!htb]
\begin{minipage}[t]{0.32\textwidth}
	\centering
	\includegraphics[trim=100 0 0 0,scale=0.4]{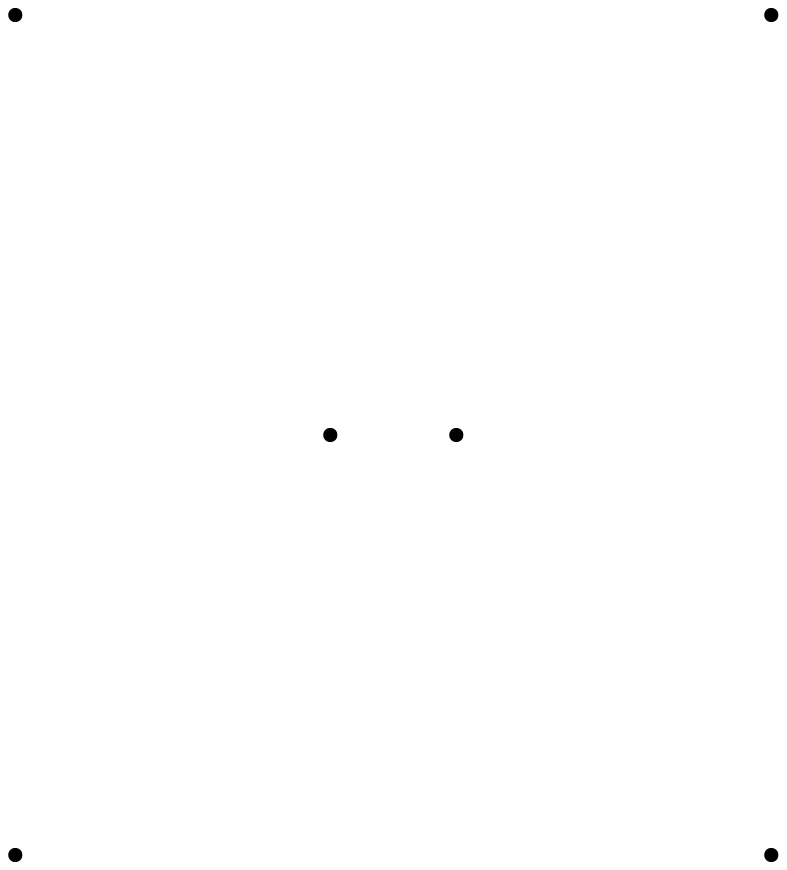}
	\caption*{$n=6, i=0, j=0$,\\ ratio $\approx 1.0238$}
	\end{minipage}
	\hfill%
	\begin{minipage}[t]{0.32\textwidth}
	\centering	
	\includegraphics[trim=100 0 0 0,scale=0.4]{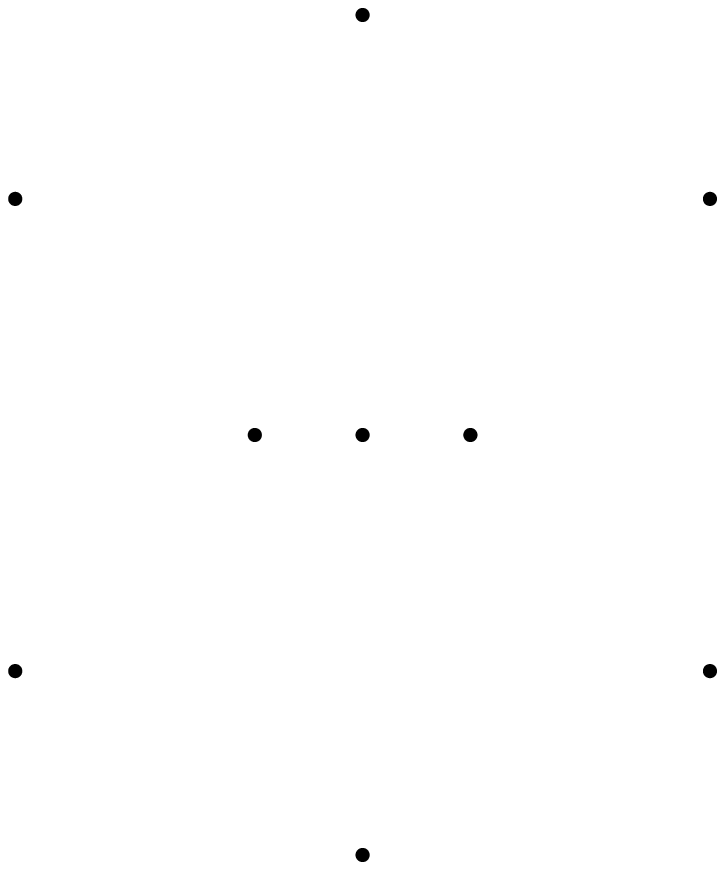}
	\caption*{$n=9, i=1, j=1$,\\ ratio $\approx 1.060$}
	\end{minipage}
  \hfill%
	\begin{minipage}[t]{0.32\textwidth}
	\centering	
	\includegraphics[trim=100 0 0 0,scale=0.4]{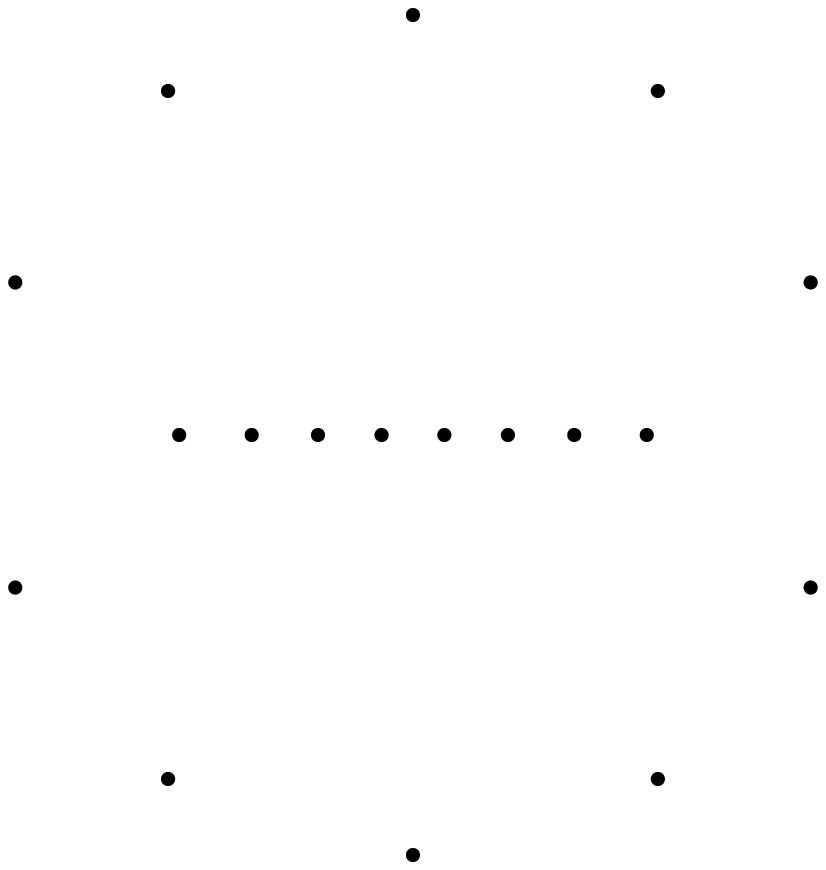}
	\caption*{$n=18, i=3, j=6$,\\ ratio $\approx 1.1319$}
	\end{minipage}
\hfill%
\begin{minipage}[t]{0.31\textwidth}
	\centering
	\includegraphics[trim=100 0 0 0,scale=0.4]{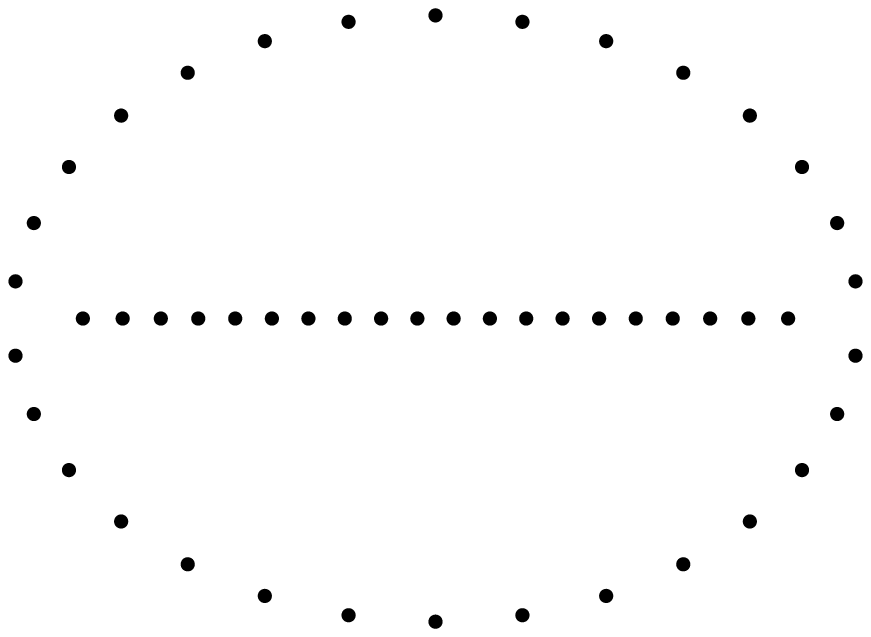}
	\caption*{$n=50, i=13, j=18$,\\ ratio $\approx 1.2263$}
	\end{minipage}
	\hfill%
	\begin{minipage}[t]{0.31\textwidth}
	\centering	
	\includegraphics[trim=100 0 0 0,scale=0.4]{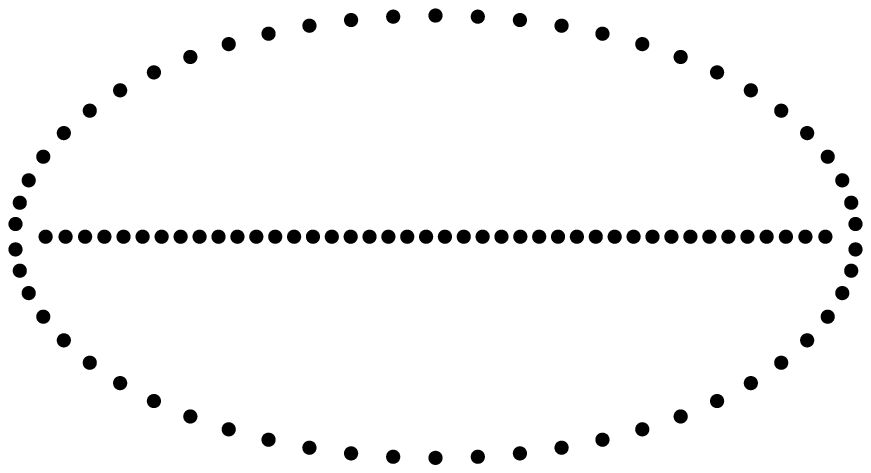}
	\caption*{$n=100, i=27, j=40$,\\ ratio $\approx 1.2695$}
	\end{minipage}
  \hfill%
	\begin{minipage}[t]{0.31\textwidth}
	\centering	
	\includegraphics[trim=100 0 0 0,scale=0.4]{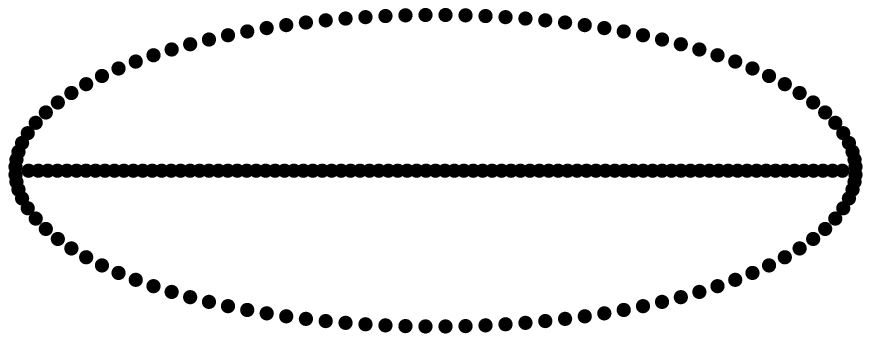}
	\caption*{$n=199, i=54, j=85$,\\ ratio $\approx 1.2970$}
	\end{minipage}
	\caption{Instances with $n$ vertices constructed by the ellipse construction algorithm.} \label{results euclidean}
\end{figure}

As in the rectilinear case we also see that vertices are not distributed equally. For large $n$ there are more inner than outer vertices. This unequal distribution occurs first at $n=18$ where we have 5+5 outer and 8 inner vertices. Moreover, for small $n$ the ellipse is nearly a circle and becomes flatter for increasing $n$.

\section{Comparing Integrality Ratio} \label{sec compare}
In this section we compare the lower bounds on the integrality ratio we found in the previous sections for the TSP variants to each other and to the instances from the literature. 

In the previous sections we showed lower bounds on the integrality ratio of $1+\frac{1}{3+2\left(\frac{5}{j+1}+\frac{1}{i+1}+\frac{1}{k+1}\right)}$ and $1+\frac{1}{3+2\left(\frac{1}{i+1}+\frac{1}{j+1}+\frac{1}{k+1}\right)}$ for the \textsc{Rectilinear} and \textsc{Multidimensional Rectilinear / Metric TSP}, respectively. As the deviation converges to 0 for $n\to \infty$, we discard in this subsection for simplicity the integrality constraints of $i,j$ and $k$. In this case the bounds for the \textsc{Rectilinear TSP} and \textsc{Multidimensional Rectilinear/ Metric TSP} are $1+\frac{1}{3+\frac{2(\sqrt{5}+2)^2}{n-3}}$ and $1+\frac{1}{3+\frac{18}{n-3}}$, respectively. As we can see, both values converge to $\frac{4}{3}$ as $n\to \infty$. By a straightforward calculation, we get
\begin{align*}
1+\frac{1}{3+\frac{2(\sqrt{5}+2)^2}{n-3}}&=1+\frac{1+\frac{2}{3}\cdot\frac{(\sqrt{5}+2)^2}{n-3}}{3+\frac{2(\sqrt{5}+2)^2}{n-3}}-\frac{\frac{2}{3}\cdot\frac{(\sqrt{5}+2)^2}{n-3}}{3+\frac{2(\sqrt{5}+2)^2}{n-3}}=\frac{4}{3}-\frac{\frac{2}{3}\frac{(\sqrt{5}+2)^2}{n-3}}{3+\frac{2(\sqrt{5}+2)^2}{n-3}}\\
&=\frac{4}{3}-\frac{\frac{2}{3}(\sqrt{5}+2)^2}{3(n-3)+2(\sqrt{5}+2)^2}.
\end{align*}
Similarly, we get $1+\frac{1}{3+\frac{18}{n-3}}=\frac{4}{3}-\frac{6}{3(n-3)+18}$. Hence, there are constants $c_1,c_2$ such that the lower bounds for the \textsc{Rectilinear} and \textsc{Multdimensonal Rectilinear / Metric TSP} are $\frac{4}{3}-\frac{\frac{2}{9}(\sqrt{5}+2)^2}{n+c_1}\approx \frac{4}{3}-\frac{3.988}{n+c_1}$ and $\frac{4}{3}-\frac{2}{n+c_2}$, respectively. Since the additive constants $c_1,c_2$ are neglectable as $n\to \infty$, we see that the latter converges to $\frac{4}{3}$ roughly twice as fast as the former.

The \textsc{Euclidean TSP} instances $G(n',\sqrt{n'-1})$ from \cite{Hougardy} have $n:=3n'$ vertices and an integrality ratio of $\frac{4n'-4+2\sqrt{n'-1}}{3n'-4+3\sqrt{n'-1}+\sqrt{n'}}\in\frac{4}{3}-\Theta(n^{-\frac{1}{2}})$ which converges asymptotically slower than that of $I_n^2$ and $I_n^3$. The tetrahedron instances for \textsc{Euclidean TSP} in \cite{Hougardy2020} have an integrality ratio between $\frac{4n+\frac{4n}{\sqrt{3}}-69}{3n+\frac{3n}{\sqrt{3}}}$ and $\frac{4n+\frac{4n}{\sqrt{3}}-17}{3n+\frac{3n}{\sqrt{3}}-33}$. These bounds are too inaccurate to directly compare the rate of convergence. Figure \ref{comp int Ratio} shows the integrality ratio which was explicitly computed in \cite{TnmIntegralityRatio}. Note that for each fixed number of vertices the tetrahedron instances depend on two parameters and they were chosen in \cite{TnmIntegralityRatio} to maximize the runtime of \texttt{Concorde} instead of the integrality ratio. 

Unfortunately, we do not have a formula for the integrality ratio of the Euclidean instances found by the ellipse construction algorithm. So we cannot directly compare their rate of convergence but only the explicitly computed integrality ratios. Figure \ref{comp int Ratio} shows the integrality ratio of the various instances described in the previous sections. Note that for the constructed Euclidean instances the integrality ratio was computed for $6\leq n\leq 109$ vertices by \texttt{Concorde}. For instances with $110 \leq n \leq 199$ vertices the runtime of \texttt{Concorde} was too high to verify the integrality ratio, see Section \ref{sec runtime concorde} for more details on this phenomenon. For these cases we rely on property \ref{prop shortest tours} of Observation \ref{property euclidean} and assume the non-intersecting shortcuts of the pseudo-tours in $\mathfrak{T}$ are optimal tours. This property holds for the instances with $6\leq n\leq 109$ vertices. As we can see from the plot, the integrality ratios of the constructed instances are higher than these of the instances $G(n',\sqrt{n'-1})$ and the tetrahedron instances. For this data they converge roughly two and four times slower than the rectilinear and metric instances, respectively.
\begin{figure}[!htb]
\centering
 \includegraphics[scale=0.7]{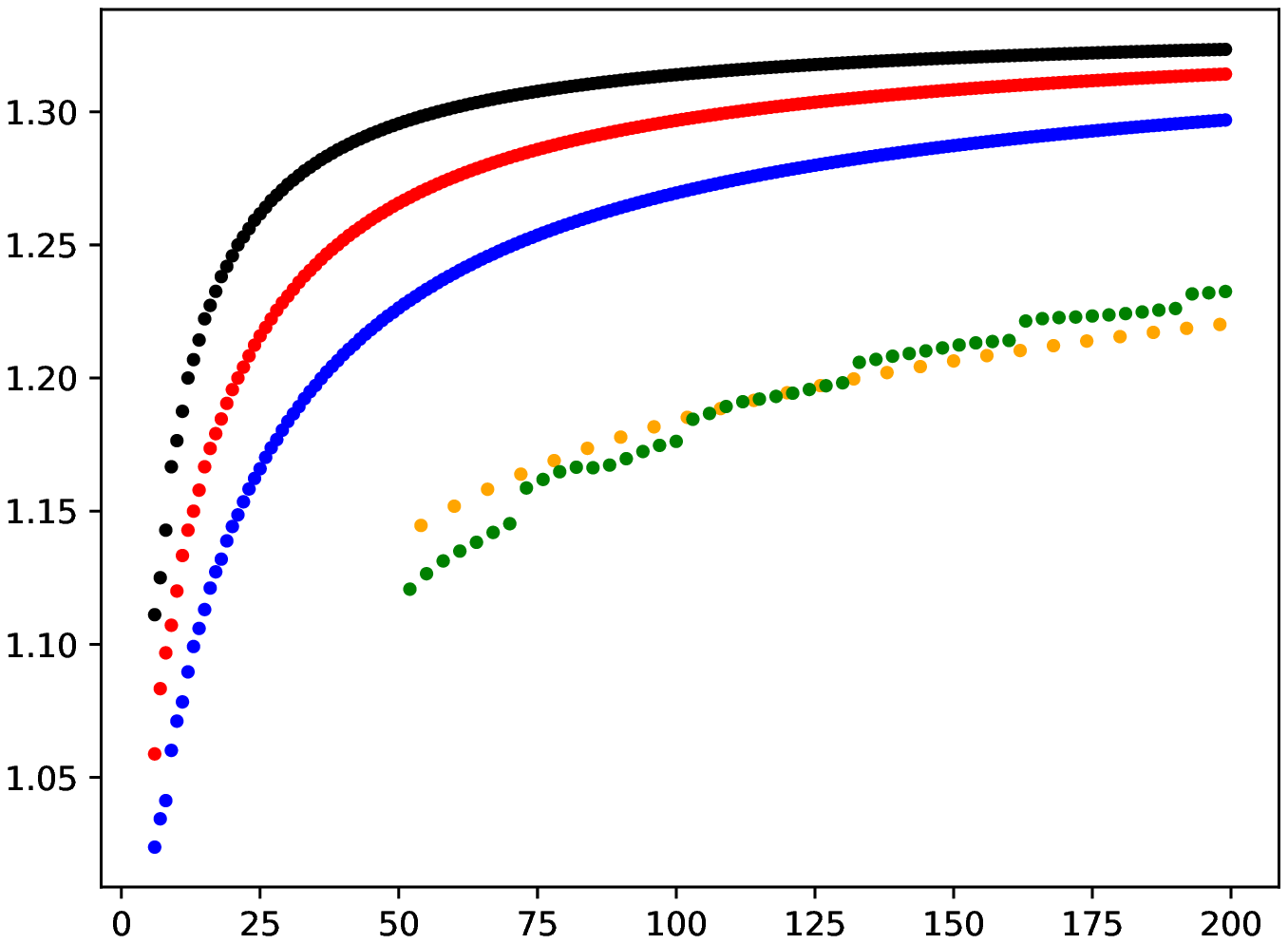}
  \caption{The integrality ratio for the TSP variants. The black, red and blue dots correspond to the lower bounds on the integrality ratio of the \textsc{Multidimensional Rectilinear TSP}/\textsc{Metric TSP} from \cite{benoit2008finding}, \textsc{Rectilinear TSP} and \textsc{Euclidean TSP}, respectively. The integrality ratio of the \textsc{Euclidean TSP} instances with $110 \leq n\leq 199$ vertices are computed assuming property \ref{prop shortest tours} of Observation \ref{property euclidean}. The orange and green dots correspond to the integrality ratios of the instances $G(n',\sqrt{n'-1})$ from \cite{Hougardy} and the tetrahedron instances from \cite{Hougardy2020}, respectively.}
  \label{comp int Ratio}
\end{figure}

\section{Hard to Solve Instances} \label{sec runtime concorde}
In this section we investigate the runtime of \texttt{Concorde} for the instances $I^3_{i,i-1,i+1}, I^3_{i,i-1,i+2}$ and $I^3_{i,i-1,i+3}$. The runtime of \texttt{Concorde} for solving these instances is much higher than for the known tetrahedron instances from \cite{Hougardy2020}.

First, we observe that symmetry seems to affect the computational results a lot. The instance $I^3_{10,10,10}$ can be solved by \texttt{Concorde} in less than a second. A small modification of the distribution of vertices on the three lines increases the runtime significantly: $I^3_{10,9,11}$ needs more than 1000 seconds to solve on the same hardware.

In the following we tested the runtime of the following instances: For $n=3(i+2)$, $n=3(i+2)+1$ and $n=3(i+2)+2$ we solved the instances $I^3_{i,i-1,i+1}, I^3_{i,i-1,i+2}$ and $I^3_{i,i-1,i+3}$ with seed 1 by \texttt{Concorde}, respectively. The distances were multiplied by 1000 and rounded down to the nearest integer. \texttt{Corcorde-03.12.19} was compiled with \texttt{gcc 4.8.5} and using \texttt{CPlex 12.04} as LP solver. We used a single core of an \texttt{AMD EPYC 7601} processor for every run. The resulting runtimes are shown in Figure \ref{runtime}. 

\begin{figure}[!htb]
\centering
 \includegraphics[scale=0.7]{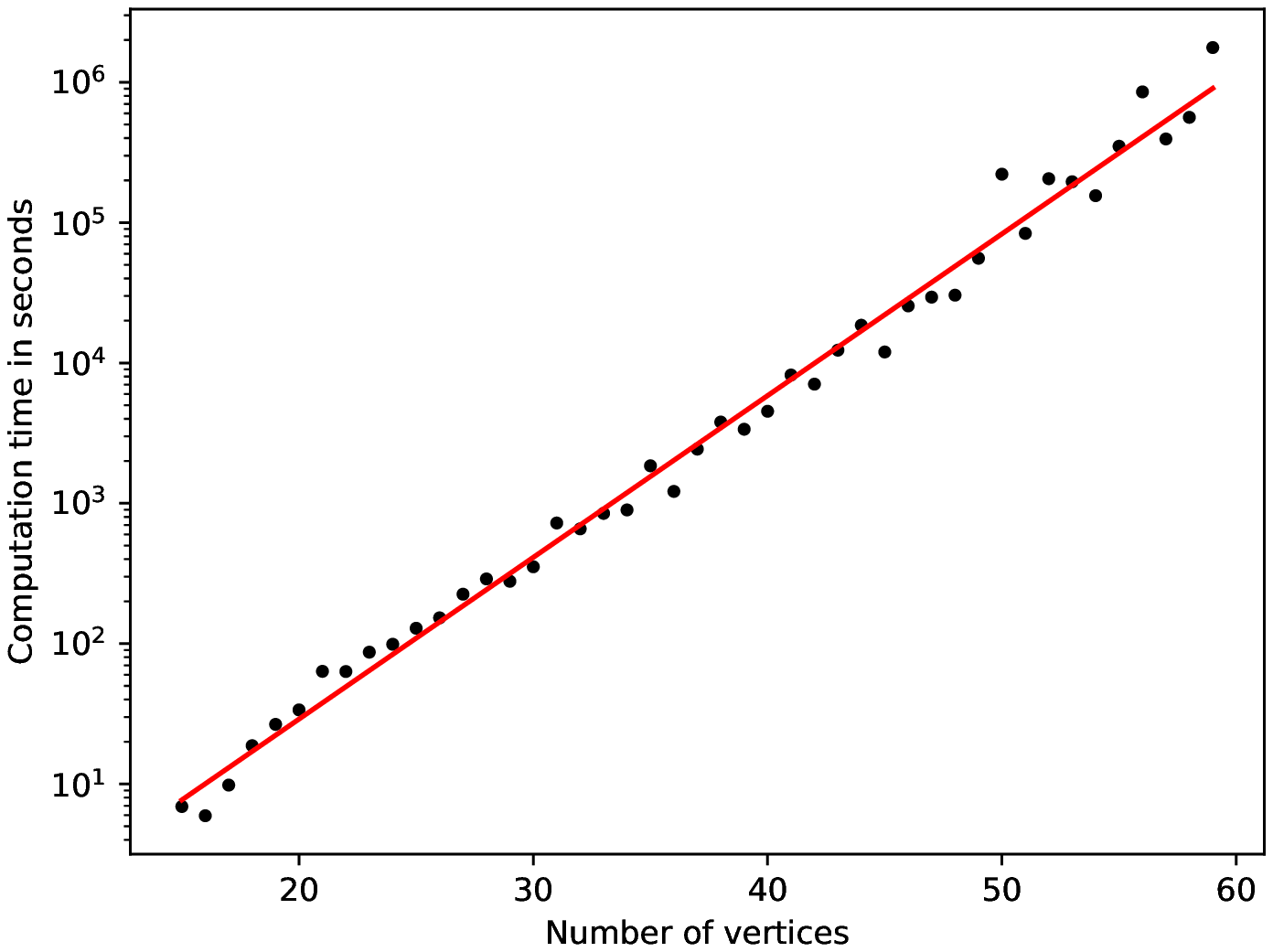}
  \caption{The \texttt{Concorde} runtime for the instances $I^3_{i,i-1,i+1}, I^3_{i,i-1,i+2}$ and $I^3_{i,i-1,i+3}$ in seconds. The red function is the least-square fit regression of the logarithmic runtime.}
  \label{runtime}
\end{figure}

Using least-square fit regression of the logarithmic runtime we get the following exponential regression for the runtime 
\begin{align*}
0.144\cdot 1.304^n.
\end{align*}

Since $1.304^3>2.2$, the estimated computation time more than doubles whenever $n$ increases by 3. Based on this estimate an instance with 100 vertices would need more than 1500 years. This runtime is much higher compared to the tetrahedron instances from \cite{Hougardy2020}. For 52 vertices the runtime for the tetrahedron instances is about 10 seconds compared to over 2 days for the tested instances. It should be noted that the calculations in \cite{Hougardy2020} were performed by a different processor which is about 20\% faster. Nevertheless, comparing with the regression function $0.480\cdot 1.0724^n$ for the tetrahedron instances we see that the exponential basis is much greater. This implies that the growth of runtime is also much faster. 

The instances $I^2_n$ seem to have similar runtimes as the instances tested above. The runtime was a bit higher if the numbers of vertices in the outer two lines are not equal. It was not obvious which distribution of vertices maximizes the runtime.

In contrast the metric instances described in \cite{benoit2008finding} with $n=36$ can be solved by \texttt{Concorde} in a few seconds compared to more than half an hour for $I^3_{10,9,11}$ with the same number of vertices. Although the runtime seems to be exponential also in this case, modifying the number of vertices on the three lines still does not give similar high runtimes. The same holds for the Euclidean instances constructed by the ellipse construction algorithm.

Note that the results of Section \ref{sec int ratio rectilinear multi} can be easily extended to determine the structure of the optimal tours of the instances $I^3_{i,i-1,i+1}, I^3_{i,i-1,i+2}$ and $I^3_{i,i-1,i+3}$. Using this knowledge the optimal tour can be computed in linear time.

\section*{Acknowledgements}
This work was supported by the Bonn International Graduate School.

\bibliographystyle{plain}
\bibliography{lower_bound_subtour_lp.bbl}

\begin{thebibliography}{10}

\bibitem{concorde}
David Applegate, Robert Bixby, Vasek Chvátal, and William Cook.
\newblock {\em Concorde-03.12.19}, 2003 (last accessed March 19, 2020).
\newblock
  \url{http://www.math.uwaterloo.ca/tsp/concorde/downloads/downloads.htm}.

\bibitem{benoit2008finding}
Genevieve Benoit and Sylvia Boyd.
\newblock Finding the exact integrality gap for small traveling salesman
  problems.
\newblock {\em Mathematics of Operations Research}, 33(4):921--931, 2008.

\bibitem{extremePoints}
Sylvia Boyd.
\newblock {\em Vertices of the Subtour Elimination Polytope}, 2010 (accessed
  March 19, 2020).
\newblock \url{http://www.site.uottawa.ca/~sylvia/subtourvertices/index.htm}.

\bibitem{BOYD2011525}
Sylvia Boyd and Robert Carr.
\newblock Finding low cost {TSP} and 2-matching solutions using certain
  half-integer subtour vertices.
\newblock {\em Discrete Optimization}, 8(4):525 -- 539, 2011.

\bibitem{boyd2010structure}
Sylvia Boyd and Paul Elliott-Magwood.
\newblock {Structure of the Extreme Points of the Subtour Elimination Polytope
  of the {STSP}}.
\newblock In {\em RIMS Kôkyûroku Bessatsu B23}, pages 33--47, 2010.

\bibitem{dantzig1954solution}
George Dantzig, Ray Fulkerson, and Selmer Johnson.
\newblock Solution of a large-scale traveling-salesman problem.
\newblock {\em Journal of the operations research society of America},
  2(4):393--410, 1954.

\bibitem{convexhull}
Vladimir~G. Deineko, René van Dal, and Günter Rote.
\newblock {The Convex-Hull-and-Line Traveling Salesman Problem: A Solvable
  Case}.
\newblock {\em {Information Processing Letters}}, pages 141--148, 1994.

\bibitem{flood1956traveling}
Merrill~M Flood.
\newblock The traveling-salesman problem.
\newblock {\em Operations research}, 4(1):61--75, 1956.

\bibitem{euclideannpg}
M.~R. Garey, R.~L. Graham, and D.~S. Johnson.
\newblock Some {NP}-complete geometric problems.
\newblock In {\em Proceedings of the Eighth Annual ACM Symposium on Theory of
  Computing}, STOC '76, pages 10--22, New York, NY, USA, 1976. ACM.

\bibitem{garey1979computers}
Michael~R Garey and David~S Johnson.
\newblock {\em Computers and intractability}, volume 174.
\newblock freeman San Francisco, 1979.

\bibitem{goemans1993survivable}
Michel~X Goemans and Dimitris~J Bertsimas.
\newblock Survivable networks, linear programming relaxations and the
  parsimonious property.
\newblock {\em Mathematical Programming}, 60(1-3):145--166, 1993.

\bibitem{grotschel1981ellipsoid}
Martin Gr{\"o}tschel, L{\'a}szl{\'o} Lov{\'a}sz, and Alexander Schrijver.
\newblock The ellipsoid method and its consequences in combinatorial
  optimization.
\newblock {\em Combinatorica}, 1(2):169--197, 1981.

\bibitem{Hougardy}
Stefan Hougardy.
\newblock On the integrality ratio of the subtour lp for euclidean {TSP}.
\newblock {\em {Operations Research Letters 42}}, pages 495--499, 2014.

\bibitem{TnmIntegralityRatio}
Stefan Hougardy and Xianghui Zhong.
\newblock {\em Hard to Solve Instances of the Euclidean Traveling Salesman
  Problem}, 2018 (accessed October 06, 2020).
\newblock \url{http://www.or.uni-bonn.de/%7Ehougardy/HardTSPInstances.html}.

\bibitem{Hougardy2020}
Stefan Hougardy and Xianghui Zhong.
\newblock Hard to solve instances of the euclidean traveling salesman problem.
\newblock {\em Mathematical Programming Computation}, 2020.

\bibitem{karamata1932inegalite}
Jovan Karamata.
\newblock Sur une in{\'e}galit{\'e} relative aux fonctions convexes.
\newblock {\em Publications de l'Institut Math{\'e}matique}, 1(1):145--147,
  1932.

\bibitem{TSPNPHARD}
Richard~M. Karp.
\newblock {\em {Reducibility among combinatorial problems in: Raymond E.
  Miller, James W. Thatcher (Eds.), Complexity of Computer Computations}}.
\newblock Plenum Press, New York, 1972.

\bibitem{monma1990minimum}
Clyde~L Monma, Beth~Spellman Munson, and William~R Pulleyblank.
\newblock Minimum-weight two-connected spanning networks.
\newblock {\em Mathematical Programming}, 46(1-3):153--171, 1990.

\bibitem{euclideannpp}
Christos~H. Papadimitriou.
\newblock The {E}uclidean traveling salesman problem is {$NP$}-complete.
\newblock {\em Theoret. Comput. Sci.}, 4(3):237--244, 1977.

\bibitem{reinelt1991tsplib}
Gerhard Reinelt.
\newblock {TSPLIB}—a traveling salesman problem library.
\newblock {\em ORSA journal on computing}, 3(4):376--384, 1991.

\bibitem{williamson1990analysis}
David~Paul Williamson.
\newblock {\em {Analysis of the Held-Karp heuristic for the traveling salesman
  problem (Master's Thesis)}}.
\newblock Massachusets Institute of Technology, 1990.

\bibitem{wolsey1980heuristic}
Laurence~A Wolsey.
\newblock Heuristic analysis, linear programming and branch and bound.
\newblock In {\em Combinatorial Optimization II}, pages 121--134. Springer,
  1980.

\end{thebibliography}
\end{document}